\documentclass{article}

\usepackage{microtype}
\usepackage{graphicx}
\usepackage{subfigure}
\usepackage{booktabs} % for professional tables

\usepackage{hyperref}

\usepackage[accepted]{icml2024}

\usepackage{amsmath}
\usepackage{amssymb}
\usepackage{mathtools}
\usepackage{amsthm}
\usepackage{stfloats}
\usepackage{wustyle}
\usepackage{pythonhighlight}
\usepackage{url}

\usepackage[capitalize,noabbrev]{cleveref}

\theoremstyle{plain}
\newtheorem{theorem}{Theorem}[section]
\newtheorem{proposition}[theorem]{Proposition}
\newtheorem{lemma}[theorem]{Lemma}

\theoremstyle{definition}
\newtheorem{definition}[theorem]{Definition}

\theoremstyle{remark}

\usepackage[textsize=tiny]{todonotes}

\icmltitlerunning{FAFE: Immune Complex Modeling with Geodesic Distance Loss on Noisy Group Frames}

\begin{document}

\twocolumn[
\icmltitle{FAFE: Immune Complex Modeling with Geodesic Distance Loss \\
on Noisy Group Frames}



\icmlsetsymbol{equal}{*}

\begin{icmlauthorlist}
\icmlauthor{Ruidong Wu}{equal,Helixon}
\icmlauthor{Ruihan Guo}{equal,Helixon}
\icmlauthor{Rui Wang}{equal,Helixon}
\icmlauthor{Shitong Luo}{Helixon}
\icmlauthor{Yue Xu}{Helixon}
\icmlauthor{Jiahan Li}{Helixon}

\icmlauthor{Jianzhu Ma}{Helixon,Tsinghua}
\icmlauthor{Qiang Liu}{Helixon,UtAustin}
\icmlauthor{Yunan Luo}{GaTech}
\icmlauthor{Jian Peng}{Helixon}

% \icmlauthor{Chenpeng Su}{Helixon}
% \icmlauthor{Rui Shen}{Helixon}
% \icmlauthor{Xiwen Zhang}{Helixon}
% \icmlauthor{Zhaohui Chen}{Helixon}

\end{icmlauthorlist}
\icmlaffiliation{Tsinghua}{Tsinghua}
\icmlaffiliation{GaTech}{GaTech}
\icmlaffiliation{UtAustin}{UTAustin}
\icmlaffiliation{Helixon}{Helixon}

\icmlcorrespondingauthor{Ruidong Wu}{ruidong@helixon.com}
\icmlcorrespondingauthor{Shitong Luo}{luost@helixon.com}
\icmlcorrespondingauthor{Jianzhu Ma}{majianzhu@tsinghua.edu.cn}
\icmlkeywords{Machine Learning, ICML, protein}

\vskip 0.3in
]



\printAffiliationsAndNotice{\icmlEqualContribution} % otherwise use the standard text.

\begin{abstract}
Despite the striking success of general protein folding models such as AlphaFold2 (AF2,~\citet{jumper2021af2}), the accurate computational modeling of antibody-antigen complexes remains a challenging task.
In this paper, we first analyze AF2's primary loss function, known as the Frame Aligned Point Error (FAPE), and raise a previously overlooked issue that FAPE tends to face gradient vanishing problem on high-rotational-error targets.
To address this fundamental limitation, we propose a novel geodesic loss called Frame Aligned Frame Error (FAFE, denoted as F2E to distinguish from FAPE), which enables the model to better optimize both the rotational and translational errors between two frames.
We then prove that F2E can be reformulated as a group-aware geodesic loss, which translates the optimization of the residue-to-residue error to optimizing group-to-group geodesic frame distance.
By fine-tuning AF2 with our proposed new loss function, we attain a correct rate of 52.3\% (DockQ $>$ 0.23) on an evaluation set and 43.8\% correct rate on a subset with low homology, with substantial improvement over AF2 by 182\% and 100\% respectively.
\footnote{Code is available at \href{https://github.com/mooninrain/FAFE.git}{https://github.com/mooninrain/FAFE.git}.}
\end{abstract}

\section{Introduction} \label{main:introduction}
Protein structure modeling is a crucial field in computational biology and has been an important unsolved problem for decades. Traditionally, structure biology researchers relied on experimental methods such as X-ray crystallography and NMR spectroscopy~\cite{wuthrich2001experiment4, jaskolski2014experiment3, bai2015experiment2, thompson2020experiment1} to determine protein structures. However, these methods are time-consuming, expensive, and highly reliant on the accessibility to pure and stable samples of target proteins. As a result, the range of protein monomers and complexes with known structures~\cite{burley2017pdb} is far limited, especially considering the tremendous growth of discovered protein sequences~\cite{uniprot2019uniprot, richardson2023mgnify} with the help of advancement in gene sequencing technologies \cite{braslavsky2003sequence, harris2008single, heather2016sequence}. This has led to an increasing interest in computational methods for protein structure prediction, which aim to predict the three-dimensional (3D) structure of proteins based solely on their sequence information.

AlphaFold2~\cite{jumper2021af2} is the most successful protein structure prediction algorithm so far. It has brought great progress in structure prediction accuracy at CASP14~\cite{kryshtafovych2021casp14}, a well-known protein structure modeling contest.
AlphaFold2-Multimer~\cite{evans2021af2multimer} was introduced subsequently by applying the framework of AF2 to multi-chain proteins.
AF2-Multimer is designed to predict massive biological complexes, especially those with certain cross-chain genetic information.
However, for complexes whose cross-chain docking poses can not be revealed by their genetic information, such as immune complexes, there is still a large gap between their predicted structures and experimental counterparts.

Immune complex modeling, which aims to model the 3D structure of an antibody-antigen complex, has important implications in antibody drug discovery~\cite{kaczor2018ppd}. Usually, the antibody part has an unknown structure as a novel sequence, and the antigen part may or may not have a corresponding experimentally determined structure, depending on different therapeutic targets. An alternative way to this problem is using a loop modeling algorithm to predict antibody structure, and applying a rigid docking algorithm between two predicted components.

Many rigid docking algorithms~\cite{pierce2014zdock,yan2020hdock,desta2020cluspro,ganea2021independent,jin2022antibody,ketata2023diffdock,wang2023injecting} have been proposed to address the antibody-antigen docking problem, bringing significant progress in recent years.
However, there are several limitations to these approaches.
First, the accuracy of such algorithms highly relies on good priors of antibody and antigen structures.
When no experimental structure is provided and noises exist in its predicted substitution, those algorithms can easily fail.
At the same time, rigid docking algorithms start with the unbound states of antibodies and antigens, thus assume no changes at the epitopes and paratopes during docking, which is not true in the real world.
This lack of flexibility can introduce errors on some targets.

To address these limitations, we tackle the immune complex modeling problem directly from the primary sequences of antibodies and antigens.
We start by analyzing AF2-Multimer as a baseline 
which uses FAPE loss as its main structure loss.
We show by derivations that cross-chain FAPE loss equals to optimizing the chordal distance between ground truth and predicted group frames on $\SE3$.
The chordal distance measures the rotational error in its chordal length, which can cause the gradient vanishing problem when the rotation angle is larger than $\frac{\pi}{2}$.
This gives us insights into why a portion of AF2-Multimer predictions are ``stuck" at wrong docking positions with large rotational errors.

With the observation above, we propose to use Frame Aligned Frame Error (F2E) which not only measures the translational errors but also the rotational errors of local frames. By adding a correcting term, F2E can be modified to approximate the geodesic distance of group frames. We then conduct extensive experiments to show the effectiveness of F2E. Overall, this paper has the following contributions:

\begin{itemize}
  \item We point out that the original FAPE loss approximates the \emph{chordal distance} loss on group frames, which can cause the gradient vanishing problem. 
  \item We propose a novel \emph{geodesic distance} loss F2E which can address the problem mentioned above.
  \item Our experiments show that the new loss can improve the antibody-antigen complex modeling performance by a large margin.
\end{itemize}

\begin{figure*}[htbp]
\vskip 0.1in
\begin{center}
\centerline{\includegraphics[width=\textwidth]{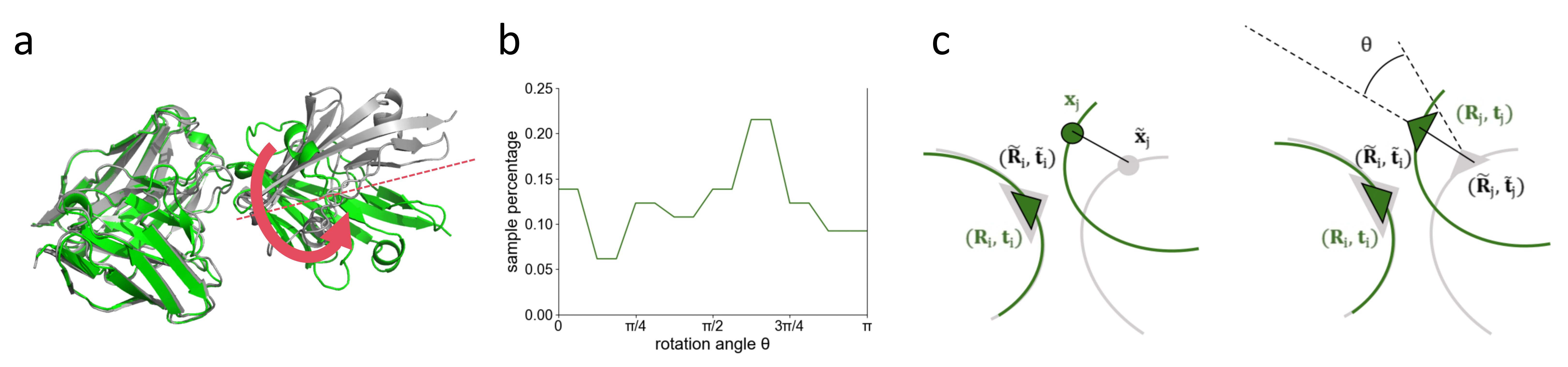}}
\caption{
Problem and method overview. (\textbf{a}) An example from PDB 7XJF, the gray structure is predicted by AF2, and the green structure is the experimental ground truth. We can see there is a large rotation error. (\textbf{b}) Statistics of rotational errors in AF2.3 predictions on our evaluation set. Most of the predictions have a rotation error larger than $\frac{\pi}{2}$. (\textbf{c}) The comparison between FAPE (left) and F2E (right). FAPE calculates point-wise Euclidean distance error after alignment to local frames, while F2E calculates frame-wise geodesic distance error after the same alignment.
}
\label{fig:teaser}
\end{center}
% \vskip -0.4in
\end{figure*}

\section{Background} \label{main:background}
In this section, we will make a brief introduction of the basic concepts and notations of protein residue frames, as well as widely used metrics to measure the distance on $\SO3$ and $\SE3$. We will then introduce the basic formula of FAPE loss.

\subsection{Frames of protein residues}
Proteins are polypeptides composed of amino acid residues. These amino acids are linked by so-called peptide bonds. Once the protein is folded, both the relative translation and rotation will be constrained by the structure to a great extent. These geometric constraints come from chemical and biological mechanisms during the formation of proteins and are crucial for proteins to function properly in organisms.

With the limited number of amino acid types, it is possible to model the structure of each amino acid residue with a certain degree of freedom, for example, adopted in the structure module of AF2~\cite{jumper2021af2}, through several frames built upon certain atoms of these residues. 

\begin{definition}
\label{def:frame}
  A frame $T = (R, \vec{t}) \in \mathrm{SE}(3)$ is a tuple of rotation matrix $R\in \mathbb{R}^{3\times 3}$ and translation vector $\vec{t}\in \mathbb{R}^{3}$.
\end{definition}

$\SE3$ is the 3D Special Euclidean Group representing all the rigid transformations, i.e., frames expressed in homogeneous coordinates.
Specifically, the rotation part $R$ belongs to 3D Special Orthogonal Group $\SO3 = \{R \in\R^{3\times 3}: R^\top R = I_3, det(R) =+1\}$, which represents the group of all 3D rotation matrices.
% \rh{Shouldn't d equal to 3?}
% \qq{we should move section 2.3 here to introduce math preliminaries.}

\subsection{Metrics on SO(3) and SE(3)}
\label{mainsub:so3_se3_metrics}
% \qq{remind what SO(3) and SE(3) are; need a place to address math background}
To measure how ``close" two frames are, we need to first define the distance metric on $\mathrm{SO}(3)$ and $\mathrm{SE}(3)$. We introduce two commonly used metrics \cite{huynh2009metrics, hartley2013rotation,carlone2022visual}, \emph{chordal and geodesic distances}.

\begin{definition}
\label{def:chordal_so3}
  Given two rotations $R_i$ and $R_j$, their \emph{chordal distance} is defined as \\
  \begin{equation}
    dist_{c}(R_i, R_j) = {\Vert R_i - R_j\Vert}_F = {\Vert R_i^\top R_j - I\Vert}_F,
  \end{equation}
\end{definition}
where $\|\cdot\|_F$ denotes the Frobenius norm.
% \rh{Attention: the last two terms are not equal. They can be bound by each other with a constant, but not equal}

\begin{definition}
\label{def:geodesic_so3}
  Given two rotations $R_i$ and $R_j$, their \emph{geodesic distance} (or angular distance) is defined as \\
  \begin{equation}
    dist_{\theta}(R_i, R_j) = \left|\arccos{\left(\frac{\Tr(R_i^\top R_j) - 1}{2}\right)}\right|,
  \end{equation}
  in which ${\Tr(\cdot)}$ denotes the matrix trace.
\end{definition}

A straightforward way to define distance metrics on $\SE3$ is to treat it as the Cartesian product of $\SO3$
and $\R^3$.

\begin{definition}
\label{def:chordal_se3}
  Given two frames $T_i=(R_i, \vec{t}_i)$ and $T_j=(R_j, \vec{t}_j)$, their \emph{chordal distance} is defined as \\
  \begin{equation}
    dist_{c}(T_i, T_j; \alpha) = \sqrt{{dist_{c}(R_i, R_j)}^2 + {\left\| \frac{\vec{t}_j - \vec{t}_i}{\alpha} \right\|}^2}
  \end{equation}
  in which $\alpha$ is an axis scaling factor.
\end{definition}

\begin{definition}
\label{def:geodesic_se3}
  Given two frames $T_i=(R_i, \vec{t}_i)$ and $T_j=(R_j, \vec{t}_j)$, their \emph{geodesic distance} (or double geodesic distance) is defined as \\
  \begin{equation}
    dist_{\theta}(T_i, T_j; \alpha) = \sqrt{{dist_{\theta}(R_i, R_j)}^2 + {\left\|\frac{\vec{t}_j - \vec{t}_i}{\alpha} \right\|}^2}.
  \end{equation}
\end{definition}

Notice that in the following sections, we do not distinguish the term ``frame" from ``pose", but there are cases where this difference matters, which are discussed in \Cref{appendixsub:delta_delta_frame}, We also would like to mention that many other valid distance metrics can be defined on $\SO3$ and $\SE3$. Some meaningful discussions are provided in \Cref{appendixsub:valid_metrics}.

\subsection{FAPE loss on a parameterized complex}
As the standard parameterization used in AF2 \cite{jumper2021af2}, proteins are modeled as a collection of $N$ amino acid residues.
Each of these residues have a backbone frame $T$ which is built upon N, C$_\alpha$ and C atoms from its backbone parts. 
The main atom coordinates can be calculated from the frame
\begin{equation}
    [N, C, C_\alpha] = T \circ [N^*, C^*, C_\alpha^*],
\end{equation}
in which $N^*, C^*, C_\alpha^*$ are the canonical coordinates of these atoms on the residue structure, and the transform $T$ is applied to individual coordinates. 
Note $C_\alpha^* = (0, 0, 0)^\top$, which means the global position of $C_\alpha$ is equal to the translation $t$ from $T = (R, t)$. These three atoms, joined with an O atom connected to C atom, constitute the basic repeating unit of proteins. The protein backbone can thus be parameterized as
\begin{equation}
    \mathbf{T} = [T_1, T_2, ..., T_N] \in \SE3^N.
\end{equation}

The sidechain part, which contains at most 10 heavy atoms and several hydrogen atoms, can also be parameterized using several sidechain frames per residue. In the paper, when we refer to residue frames, we are mentioning backbone frames if not specified. However, we do hope to point out that any geometric derivations with backbone frames are also applicable to sidechain frames, as there is no difference between them from a pure geometric perspective.

Specifically, a protein complex is a collection of multiple protein chains, usually denoted in capital letters $[A, B, C...]$. 
A protein chain is a collection of residues sequentially connected by peptide bonds.
A protein complex can therefore be parameterized as
\begin{equation}
    \mathbf{T} = [\mathbf{T}_A, \mathbf{T}_B, \mathbf{T}_C, ...],
\end{equation}
in which $\mathbf{T}_A$ represents the collection of all backbone frames on protein chain A.

% \qq{Need a more systematic introduction of FAPE here. 
% We can start from a formal description of how each protein structure is represented, and how the overall FAPE loss is defined in a bottom up way.}
FAPE is a loss function originally proposed by AF2~\cite{jumper2021af2}. Its basic principle is to estimate the point-wise error after alignment of each local frame of the proteins, as shown in \Cref{fig:teaser} (\textbf{c}). FAPE loss allows the calculation of residue-atom pairwise error with the good property of $\SE3$ invariance, that is, the loss function is invariant to transformations of predicted or ground truth structure.

\begin{definition}
\label{def:backbone_fape}
  Given a ground truth complex structure $\mathbf{T} = [T_1, T_2, ... T_N] \in \SE3^N$ and its predicted counterpart $\hat{\mathbf{T}} = [\hat{T_1}, \hat{T_2}, ... \hat{T_N}] \in \SE3^N$, The \emph{backbone FAPE loss} is defined as
  \begin{align}
    & L_{\mathrm{FAPE}}(\mathbf{T}, \hat{\mathbf{T}}) = \frac{1}{(N-1)N} \times \\
    & \sum\limits_{i, j \in [1, N], i\neq j}{\left\|\hat{T_i}^{-1} \circ \hat{\vec{t}}_j - {T_i}^{-1} \circ \vec{t}_j \right\|}_F ,
  \end{align}
  in which $T_j = (R_j, \vec{t}_j)$ and $\hat{T}_j = (\hat{R}_j, \hat{\vec{t}}_j)$.
\end{definition}

\begin{definition}
\label{def:inter_chain_fape}
Given a ground truth two-chain complex structure $\mathbf{T} = [\mathbf{T}_A, \mathbf{T}_B]$ and its predicted counterpart $\hat{\mathbf{T}} = [\hat{\mathbf{T}}_A, \hat{\mathbf{T}}_B]$, the \emph{inter-chain backbone FAPE loss} is defined as
  \begin{align}
    & L_{\mathrm{FAPE}}(\mathbf{T}, \hat{\mathbf{T}}) = \frac{1}{N_AN_B} \times \\
    & \sum\limits_{\substack{i \in [1, N_A] \\ j \in [N_A + 1, N_A + N_B]}}{\Vert \hat{T_i}^{-1} \circ \hat{\vec{t}}_j - {T_i}^{-1} \circ \vec{t}_j \Vert}_F,
  \end{align}
  in which residues $[1, N_A]$ belong to chain A and $[N_A+1, N_A+N_B]$ belong to chain B. The definition on more than 2 chains is trivial by weighted averaging over all chain pairs.
\end{definition}
For notation simplicity, in the following words, we use $L_{\mathrm{FAPE}}(i, j)$ to denote the term ${\Vert \hat{T_i}^{-1} \circ \hat{\vec{t}}_j - {T_i}^{-1} \circ \vec{t}_j \Vert}_F$. We also use notation $i\in \{chainA\}$ and $j\in \{chainB\}$ to indicate residue $i$ belongs to chain A and residue $j$ belongs to chain B, and operator $chain(\cdot)$ to indicate finding the chain of the given residue.

A detailed introduction of FAPE is in \Cref{appendix:loss}.

\section{FAPE is a chordal distance estimation of group frames} \label{main:fape}
With the defined metrics above, in this section, we discuss the properties of the original FAPE loss adopted in AF2 and AF2-Multimer~\cite{jumper2021af2,evans2021af2multimer}. We show that with small modifications, FAPE loss can be regarded as a chordal distance estimation of group frames.

\subsection{Definition of group frame and its distance}
\label{mainsub:def_group_frame}
First of all, we need to define the group frame and its metric. This is important because when we work on complex modeling, we care more about inter-chain errors than intra-chain errors. The inter-chain errors, which we will show in the discussions below, mainly come from the incorrect relative pose estimation of different chains.

However, defining a group frame on a non-rigid protein group is not a trivial problem. In this paper, we avoid dealing with this problem directly. Instead, we consider defining the frame differences rather than the frames themselves.

Inspired by the approach in ligand RMSD \cite{mendez2003ligandrmsd}, we can define the distance of frames between predicted structures and ground-truth ones by superposition.

\begin{definition}
\label{def:delta_frame}
  Given any frame $T$ in the protein structure and its predicted counterpart $\hat{T}$, $T, \hat{T}\in \SE3$, their distance is defined as
  {\setlength\abovedisplayskip{0.3cm}
  \setlength\belowdisplayskip{0.3cm}
  \begin{equation}
    dist(\hat{T}^{-1}, T^{-1}) = f(\hat{T} T^{-1}) = f(T^{\mathrm{align}}),
  \end{equation}}
  in which $f(\cdot)$ is the chordal or geodesic distance function and $T^{\mathrm{align}}$ is defined as the rigid transformation calculated for superposition of atoms on the frame. The proof for the equation is given in \ref{appendix:sub:def_compose}.
  % \qq{some minor gap in definition here, distances are two variable functions}
\end{definition}
The definition above allows us to calculate the group frame difference when the group is non-rigid, which is important in immune complex modeling without ground truth structures. Next, we can calculate the relative pose error of two chains in the form of ``diff-diff-frame".
\begin{definition}
\label{def:delta_delta_frame}
  Given any frame $T_A$ and $T_B$ on the protein structure, we have $T_A, T_B, \hat{T}_A, \hat{T}_B\in \SE3$, we define their relative pose distance error as
  {\setlength\abovedisplayskip{0.3cm}
  \setlength\belowdisplayskip{0.3cm}
  \begin{equation}
    \begin{aligned}
       error(A, B;\alpha) = & \; dist({\hat{T}_A}^{-1} \hat{T}_B, {T_A}^{-1} T_B; \alpha) \\
                   = & \; dist(\hat{T}_B, T_A^{align} T_B;\alpha).
    \end{aligned}
  \end{equation}
  }
  The distance could either be chordal distance or geodesic distance. See detailed discussion in \Cref{appendixsub:delta_delta_frame}.
\end{definition}
The definition above makes sure that the relative pose error only relies on the superposition of two chains separately regardless of any specific definition of group frames themselves. 

\subsection{Group-aware FAPE}
\label{mainsub:g-fape}
In this part, we are going to show that by modifying the original FAPE loss, we can obtain a group-aware version of FAPE, which we shorten by G-FAPE. We show that applying G-FAPE on residue-point pairwise constraints between groups is equivalent to directly optimizing the chordal distance between rigid group frames.
% \qq{the derivation is long and mess here. can we write the main results in a concise and readable way, or a theorem, and push the detailed derivation into a proof?}

For $i\in\{chain A\}$ and $j\in\{chain B\}$, we denote $T_A$ and $T_B$ as the group frame of chain A and B, $\vec{x}_{i,local}$ as the local position of the C-alpha atom $i$ at $T_A$. Specifically, we build $T_A$ with its origin at the Euclidean average of all points $\vec{x}_i, i\in\{chain A\}$, which implies
\begin{equation}
  \sum\limits_{i\in\{chain A\}} \vec{x}_{i,local} = 0.
\end{equation}
In addition, for the predicted structures, we assume a Gaussian noise $\vec{\epsilon}_i$ added to each $\vec{x}_{i,local}$. It is then straightforward to write the FAPE loss as
\begin{equation}
  \begin{aligned}
  &    L_{\mathrm{FAPE}}(i, j) = \Vert  ({T_j}^{-1} T_A) \circ \vec{x}_{i, local} \\
  & \; -{({\hat{T}_j}^{-1} \hat{T}_A) \circ (\vec{x}_{i, local} + \vec{\epsilon}_i) \Vert}_F.
  \end{aligned}
\end{equation}
Define $T_{j,local}=T_A T_j$, in which $T_A$ is the group frame of chain A. We then introduce the \emph{near-rigid} condition, that is, the local frame error is relatively small compared to the group frame error, and can be ignored. Empirically, we can see from \Cref{fig:teaser} (a) that most local frames satisfy this condition. So we have
\begin{equation}
  \begin{aligned}
    % {\hat{T)_j}^{-1} \hat{T)_B = & {\hat{T}_{j,local}}^{-1} ({\hat{T}_A}^{-1} \hat{T}_B) \\
    %                  \approx & {T_{j,local}}^{-1} ({\hat{T}_A}^{-1} \hat{T}_B).
  {\hat{T}_j}^{-1} \hat{T}_A = & {\hat{T}_{j,local}}^{-1} ({\hat{T}_B}^{-1} \hat{T}_A) \\
                     \approx & {T_{j,local}}^{-1} ({\hat{T}_B}^{-1} \hat{T}_A).
  \end{aligned}
\end{equation}
This helps us to simplify FAPE loss as
\begin{equation}
  \begin{aligned}
    &    L_{\mathrm{FAPE}}(i, j) = \Vert  ({T_B}^{-1} T_A) \circ \vec{x}_{i, local} \\
  & \; -{({\hat{T}_B}^{-1} \hat{T}_A) \circ (\vec{x}_{i, local} + \vec{\epsilon}_i) \Vert}_F.
  % &    L_{\mathrm{FAPE}}(i, j) = \Vert  ({\hat{T}_A}^{-1} \hat{T}_B) \circ \vec{x}_{i, local} \\
  % & \; -{({T_A}^{-1} T_B) \circ (\vec{x}_{i, local} + \vec{\epsilon}_i) \Vert}_F.
  \end{aligned}
\end{equation}

To further simplify, we denote $(R_{\Delta}, \vec{t_{\Delta}}) = T_{\Delta}=T_B^{-1} T_A$. Now we are ready to introduce G-FAPE, which is simply the quadratic mean of FAPE loss over $i \in \{chain A\}$.
\begin{equation}
\label{eq:g_fape}
  \begin{aligned}
    & L_{\mathrm{G-FAPE}}(j) = \sqrt{\sum\limits_{i \in \{chain A\}} L^2_{\mathrm{FAPE}}(i, j)} \\
  = & \bigg[\sum\limits_{i \in \{chain A\}}\vec{\epsilon}_i^{\top} \vec{\epsilon}_i \\
    & + \Big(\sum\limits_{i \in \{chain A\}}\vec{x}_{i,local}^{\top} \vec{x}_{i,local}\Big) 
    {\Vert R_{\Delta} - {\hat{R}_{\Delta} \Vert}}_F^2 \\
    & + n {\Vert \vec{t}_{\Delta} - \hat{\vec{t}}_{\Delta} \Vert}_F^2\bigg]^{\frac{1}{2}} \\
  = & \sqrt{\sum\limits_{i \in \{chain A\}}\vec{\epsilon}_i^{\top} \vec{\epsilon}_i + 
  k^2 \cdot dist_{c}^2(\hat{T}_{\Delta}, T_{\Delta}; \alpha)},
  \end{aligned}
\end{equation}
in which $k^2=\sum\limits_{i \in \{chain A\}}\vec{x}_{i,local}^{\top} \vec{x}_{i,local}$ and $\alpha^2=\frac{k^2}{n}$, $n=card(\{\text{C-alpha atoms on chain A}\})$. The term $\sum\limits_{i \in \{chain A\}}\vec{\epsilon}_i^{\top} \vec{\epsilon}_i$ represents the intra-chain noises that the inter-chain loss can not optimize. Thus, we can conclude that optimizing G-FAPE is equivalent to optimizing the chordal error of predicted and ground truth $T_{\Delta}$.

\subsection{Gradient vanishing problem in G-FAPE}
\label{mainsub:vanishing}
The chordal distance is a valid metric in many circumstances, but it can cause problems in complex modeling problems. As the \Cref{fig:teaser} (a, b) show, most of the AF2 predictions on our evaluation set have a high rotation error (larger than $\frac{\pi}{2}$).

To understand this phenomenon, in \Cref{fig:gradient} we draw the loss curve of chordal distance to the angle $\theta$ of rotation error, which is calculated according to \Cref{def:geodesic_so3}. We can see the gradient of chordal distance loss is gradually decreasing when $\theta$ approaches $\pi$. As a result, models trained with FAPE loss will overly concentrate on samples with a small rotation error at the start--for example, samples whose MSAs provide enough contact information to determine the docking pose roughly. Harder samples, like antibody-antigen pairs, are relatively under-trained.

\begin{figure}[th]
\vskip -0.1in
\begin{center}
\centerline{\includegraphics[width=\columnwidth]{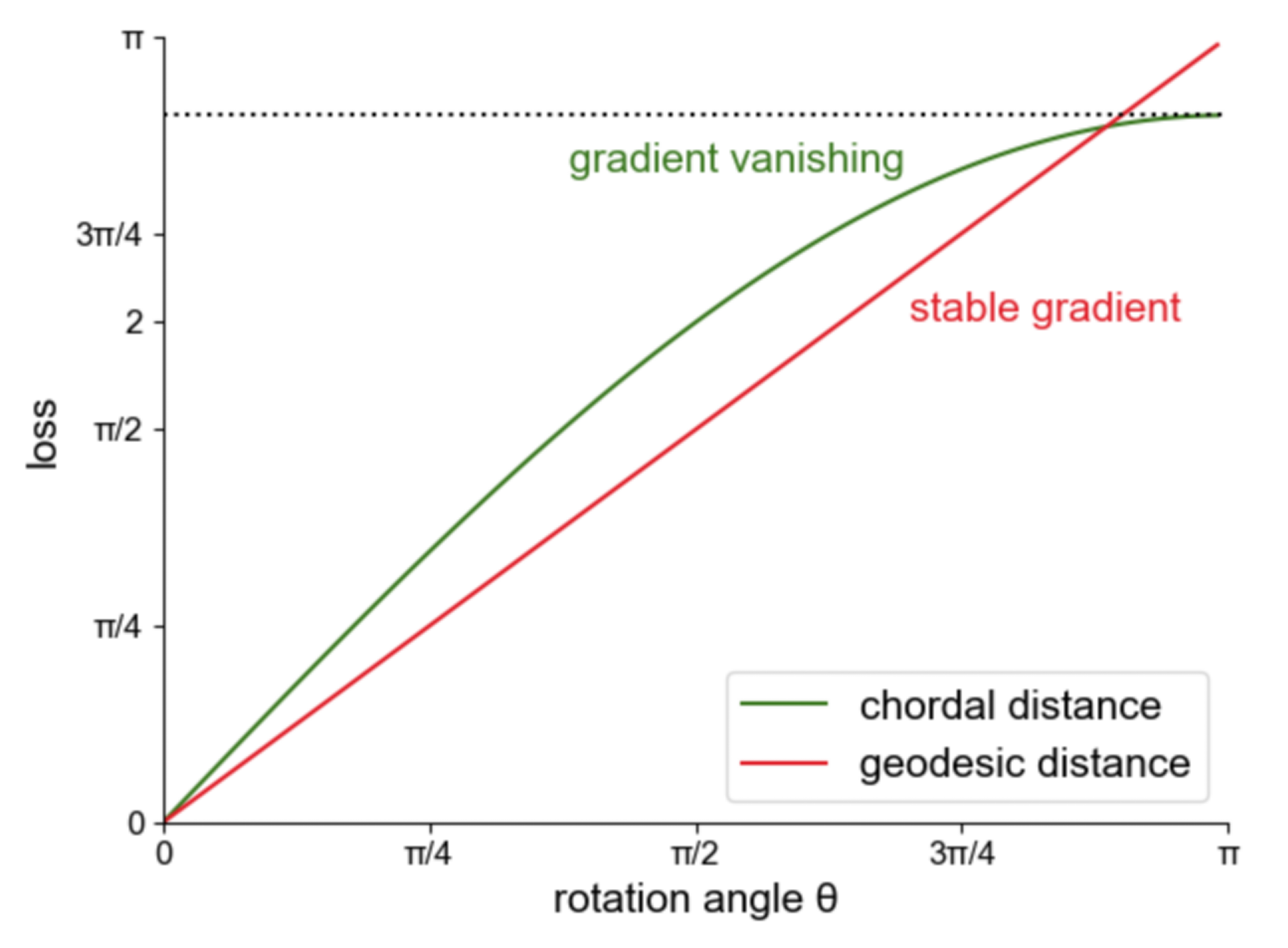}}
\caption{
The loss value regarding rotation angle (difference between predicted and groundtruth) for chordal and geodesic distance. The gradient of chordal distance tends to vanish when the rotation angle approaches $\pi$, while the geodesic distance provides a stable gradient from $0$ to $\pi$. 
}
\label{fig:gradient}
\end{center}
\vskip -0.4in
\end{figure}

\section{Using F2E to estimate geodesic distance of group frames} \label{main:f2e}
With the gradient vanishing problem mentioned above, we propose to develop new loss functions to address it. One straightforward idea is to directly optimize $dist_{\theta}(T_{\Delta}, \hat{T}_{\Delta})$. However, this leads to several difficulties. First of all, according to \Cref{def:delta_frame} and \Cref{def:delta_delta_frame}, the calculation of $T_{\Delta}$ and $\hat{T}_{\Delta}$ relies on superposition, which means we need to calculate the gradient of Singular Value Decomposition (SVD), which is numerically unstable. Indeed, when chains are nearly rigid, the singular values will get close and make the gradient unstable. Another potential approach is to estimate $T_{\Delta}$ directly through a separate neural network. We do not use this approach because we want our methods to be compatible with the existing architecture of AF2.

\subsection{From FAPE to F2E}
We propose to transform the original FAPE loss into a geodesic loss that can keep the gradients stable throughout different rotation angles. To achieve this, we propose Frame Aligned Frame Error (F2E), which measure the geodesic distance between two aligned frames.

\begin{definition}
\label{def:f2e}
  Given any residue frame $i$ and any residue frame $j$ on the protein structure, we have $\hat{T}_i, T_i, \hat{T}_j, T_j\in \mathrm{SE}(3)$, the F2E loss is defined as
  \begin{equation}
    L_{\mathrm{F2E}}(i, j; \alpha) = dist_{\theta}({\hat{T}_i}^{-1} \hat{T}_j, {T_i}^{-1} T_j; \alpha),
  \end{equation}
  in which $dist_{\theta}(\cdot)$ denotes the geodesic distance defined in \Cref{def:geodesic_se3}.
\end{definition}
With denotion $(R_{i\rightarrow j}, t_{{i\rightarrow j}})=T_{i\rightarrow j}=T_i^{-1}T_j$, we further show that F2E loss can degenerate to FAPE loss, as
\begin{equation}
\label{eq:f2e_degenerate_to_fape}
  L^2_{\mathrm{F2E}}(i, j; \alpha) = dist^2_{\theta}(\hat{R}_{i\rightarrow j}, R_{i\rightarrow j}) + \frac{L^2_{\mathrm{FAPE}}(i, j)}{\alpha^2}.
\end{equation}

\subsection{Group-aware F2E}
We first look into what will happen if we apply F2E in \Cref{def:f2e}. With \Cref{eq:g_fape} and \Cref{eq:f2e_degenerate_to_fape}, we can get
\begin{equation}
  \begin{aligned}
      & \; \sum\limits_{i \in \{chain A\}} L^2_{\mathrm{F2E}}(i, j) \\
    = & \; \frac{L^2_{\mathrm{G-FAPE}}(j)}{\alpha^2} + \sum\limits_{i \in \{chain A\}} dist^2_{\theta}(\hat{R}_{i\rightarrow j}, R_{i\rightarrow j}) \\
\approx & \; \frac{L^2_{\mathrm{G-FAPE}}(j)}{\alpha^2} + n \cdot dist^2_{\theta}(\hat{R}_{\Delta}, R_{\Delta}) \\
    = & \; \frac{\sum\limits_{i \in \{chain A\}}\vec{\epsilon}_i^{\top} \vec{\epsilon}_i + 
  k^2 \cdot dist_{c}^2(\hat{T}_{\Delta}, T_{\Delta}; \frac{k}{\sqrt{n}})}{\alpha^2} \\ 
      & \; + n \cdot dist^2_{\theta}(\hat{R}_{\Delta}, R_{\Delta}) \\
    = & \; \frac{\sum\limits_{i \in \{chain A\}}\vec{\epsilon}_i^{\top} {\vec{\epsilon}_i}}{\alpha^2} + \frac{n}{\alpha^2}{\left\|\hat{\vec{t}}_{\Delta} - \vec{t}_{\Delta} \right\|}^2 \\
    + & \;\frac{k^2}{\alpha^2} \cdot dist_{c}^2(\hat{R}_{\Delta}, R_{\Delta}) + n \cdot dist^2_{\theta}(\hat{R}_{\Delta}, R_{\Delta}). \\
  \end{aligned}
\end{equation}
Notice that we get both geodesic and chordal distance error in the result. Fortunately, the chordal error can be inferred from the geodesic error.
\begin{lemma}
\label{prop:geo2chordal}
    Given two rotations $R_i$ and $R_j \in \SO3$, we have
    \begin{equation}
        d_{c}(R_i, R_j) = 2\sqrt{2}\sin(\frac{d_{\theta}(R_i, R_j)}{2}),
    \end{equation}
    in which $d_c(\cdot)$ denotes the chordal distance defined in \Cref{def:chordal_so3} and $d_{\theta}(\cdot)$ denotes the geodesic distance defined in \Cref{def:geodesic_so3}. The proof is provided in \Cref{proof:geo2chordal}.
\end{lemma}
So we can actually "cancel" the chordal error term by introducing group-aware F2E (G-F2E) as
\begin{definition}
\label{def:gf2e}
  Given any residue frame $i$ and any residue frame $j$ on the protein structure, $chain(i) \neq chain(j)$, G-F2E loss is defined as
  \begin{equation}
    \begin{aligned}
    & L_{\mathrm{G-F2E}}(i, j; \alpha) = \bigg[\frac{L^2_{\mathrm{FAPE}}(i, j)}{\alpha^2} \\
    & + \theta^2 - \frac{8\vec{x}_{i,local}^{\top} \vec{x}_{i,local}}{\alpha^2} \cdot \sin^2\frac{\theta}{2}\bigg]^{\frac{1}{2}},
    \end{aligned}
  \end{equation}
  in which $\theta=dist_{\theta}(\hat{R}_{i\rightarrow j}, R_{i\rightarrow j})$, and $\vec{x}_{i,local}$ is the coordinates at the group frame on $chain(i)$. We repeat that the group frame is built with its origin at the Euclidean average point of all C-alpha atoms, and without assumptions about its rotation.
\end{definition}

\section{Experiments} \label{main:exp}
We analyze the effectiveness of our proposed loss by fine-tuning from original AF2-Multimer weights.
\subsection{Parameter-efficient fine-tuning of AF2-Multimer}

\begin{figure*}[htbp]
\vskip 0.1in
\begin{center}
\centerline{\includegraphics[width=0.8\textwidth]{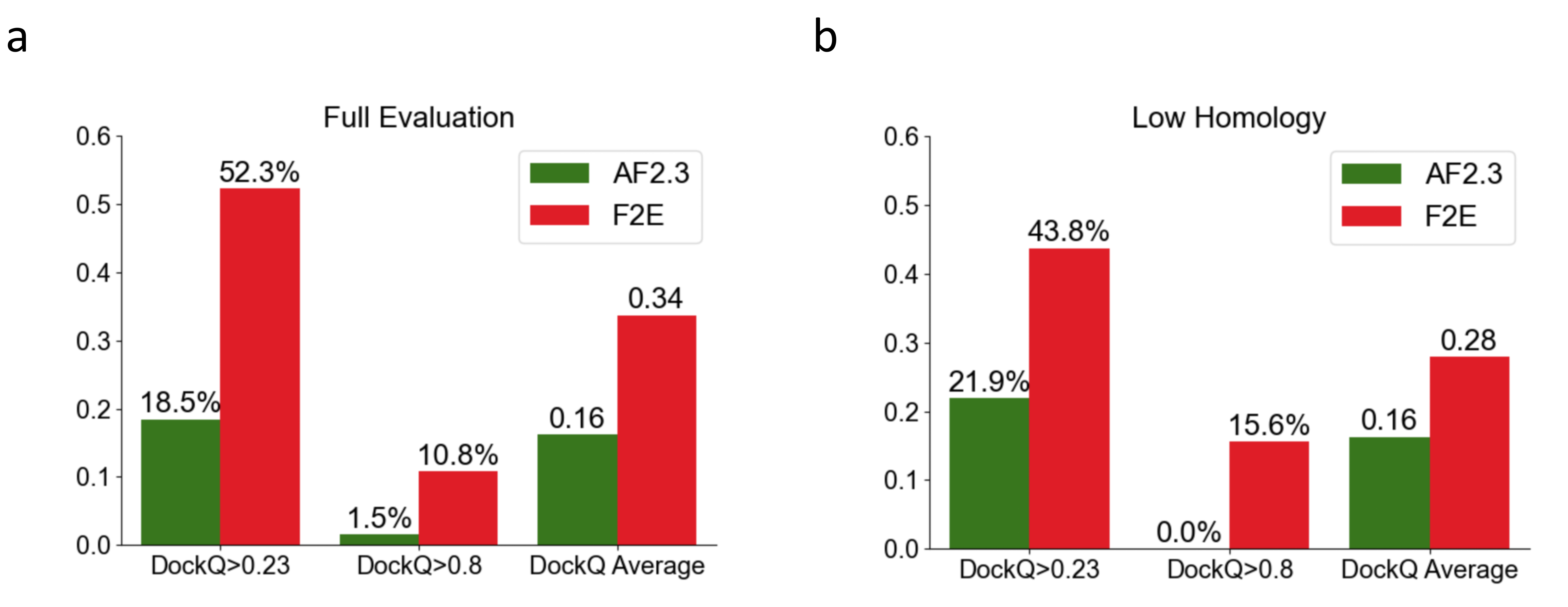}}
\caption{
Comparison of F2E with AF2.3 on two held evaluation sets. (\textbf{a}) The evaluation results on full evaluation set. For columns ``DockQ$>$0.23" and ``DockQ$>$0.8" the percentage is reported, and average DockQ is reported in the last column. (\textbf{b}) The evaluation results on low homology set with novel antigens.
}
\label{fig:main_exp}
\end{center}
\vskip -0.4in
\end{figure*}
We conduct experiments by fine-tuning the original AF2 weights. The fine-tuning dataset is collected from The Structural Antibody Database (SabDab)~\cite{dunbar2014sabdab}, which provides annotations on the original PDB database~\cite{burley2017pdb} raw structures. The training set only includes samples released before September 30th, 2021, which is the same cut off for AF2.3 training set. More data preparation details are available in \Cref{appendix:sub:detail_data}.

In the early version of F2E, we tried both to finetune the full model and to only finetune the structure module.
Fine-tuning the full model leads to a severe overfitting problem, where we find the accuracy of training samples arise fast while the evaluation performance drops very soon after a small rise, especially on LDDT, which suggests a memory loss of the intra-chain structure information.
while the structure modules finetuning only leads to a small increase in performance. 
We therefore introduce an intermediate solution with Low-Rank Adaption (LoRA)~\citep{hu2022lora}.
Note that we use LoRA not because of the computational resources exceed our capacity, but because it provides a natural path to avoid catastrophic forgetting.

We search and discover a relatively good LoRA setting for fine-tuning AF2, where the LoRA layers are only added to 48-layer evoformer, with other weights freezed. All the weights before the evoformer like template embedder and extra MSA stack are freezed, and all the weights after the evoformer including the structure module and all the prediction heads are fully trainable. More training details can be found in \Cref{appendix:sub:detail_train}.

Following the standard AF2 settings, we fine-tune 5 models and report the performance of top1 prediction with iptm as the confidence score. A discussion of inference details is provided in \Cref{appendix:sub:detail_inference}.

\subsection{Results}
In \Cref{fig:main_exp} we show the experimental results on two held evaluation datasets. 
% We prepared two evaluation datasets for benchmark. 
Our protocol follows in principle that in AlphaFold-latest (AF3)~\cite{2023af2latest}. 
We first filter a full evaluation dataset with all the samples after the training cut off, getting 693 unique samples.
We use MMseqs2~\cite{steinegger2017mmseqs2} to cluster all the antigen sequences from both training and evaluation set.
For the full evaluation set, we select all the cluster centers.
For low homology evaluation set, we discard all the data samples whose antigens belong to a cluster with training samples.

All the results reported in \Cref{fig:main_exp} are top1 predictions across 5 models plus 5 seeds with iptm as the confidence score. 
We use DockQ~\cite{basu2016dockq} as the evaluation metric for antibody-antigen complex modeling.
As commonly suggested~\cite{2023af2latest}, DockQ$>$0.23 is considered accurate predictions, while DockQ$>$0.8 is consider high accuracy.
The experimental results show that F2E is able to increase the accuracy in the definition of DockQ$>$0.23 by a large margin, bringing 182\% increase on the full evaluation set and 100\% increase on the low homology set.
In addition, our fine-tuned model is able to predict a number of samples in high accuracy, which is rare for original AF2.3 predictions.
We notice the impressive performance of AF3~\citep{2023af2latest}, but we can not build a fair benchmark with it because AF3 predicts the structure of the whole complex with the help of small molecules, peptides and other protein chains. 
We expect to extend our results on newer protein foundation models with full-atom prediction ability in the future.

\subsection{Ablation tests}
To further analyze the effectiveness of different components of our approach, we conduct ablation tests for different settings with results shown in \Cref{table:ablation}.
All results in the table is reported by fine-tuning the same AF2 model and inference with 5 random seeds.
The only exception is the top row with ensemble on, which suggests the most-confident result of 5 models from the same protocol in \Cref{fig:main_exp}.
The bottom baseline is to inference with the original model weights.
We also add the baseline of AF2sample~\cite{wallner2023af2sample}.
AF2sample is an agressive sampling protocol using AF2.1 and AF2.2 with multiple strategies, including removing templates, re-sampling MSAs, increasing the recycling rounds and adding dropout during inference time.
More than 200 samples are proposed for each target.
We think AF2sample represents the full potential of AF2 without fine-tuning. 
%and note that it is the only method with such a sophisticated sampling in ~\Cref{table:ablation}.

We can see that single-model F2E is able to reach the same performance of AF2sample, and a simple ensemble with 25 samples in total already surpasses it, suggesting a non-trivial improvement.
% Other than that, we find several messages from \Cref{table:ablation}.
Additionally we can conclude from \Cref{table:ablation}: 
1) group-aware adaptation of FAPE is not only theoretically sound, but also brings an increase in performance.
2) F2E further increases DockQ compared to G-FAPE, suggesting the effectiveness of stablizing the gradient across rotation errors.
3) The ensemble technique mainly improves the average DockQ (from 0.210 to 0.279) of existing successful predicitions but does not improve the percentage of success as much (from 0.375 to 0.438).
Moreover We find that the best predictions according to ground truth DockQ can reach a high accuracy rate of 0.59. This suggests necessity for further improvement in the confidence score.

\begin{table}[htbp]
\vskip -0.1in
\caption{Ablations study. The loss with "AF2" means no fine-tuning, "AF2S" means AF2sample. "G-aware" indicates whether to use the group-aware version of the loss and "ensemble" means top1 prediction across 5 models with iptm as the confidence score.}
\label{table:ablation}
% \vskip 0.15in
\begin{center}
\resizebox{0.9\linewidth}{!}{
\begin{tabular}{lccc|cc}
\toprule
Loss & LoRA & G-aware & Ensemble & $\mathrm{DockQ>0.23}$ \\
\midrule
F2E      &  $\surd$  & $\surd$ & $\surd$ & \textbf{0.438} \\  % 0.279
F2E      &  $\surd$  & $\surd$ & $\times$ & 0.375 \\  % 0.210
FAPE     &  $\surd$  & $\surd$ & $\times$ & 0.313 \\  % 0.188
FAPE     &  $\surd$  & $\times$ & $\times$ & 0.250 \\  % 0.149
FAPE     &  $\times$  & $\times$ & $\times$ & 0.219 \\  % 0.134
\midrule
AF2S &  $\times$  & $\times$ & $\times$ & 0.375 \\  % 0.119
AF2 &  $\times$  & $\times$ & $\times$ & 0.188 \\  % 0.119
\bottomrule
\end{tabular}
}
\end{center}
\vskip -0.2in
\end{table}

\subsection{Case study and error analysis}
\begin{figure*}[htbp]
% \vskip -0.2in
\begin{center}
\centerline{\includegraphics[width=\textwidth]{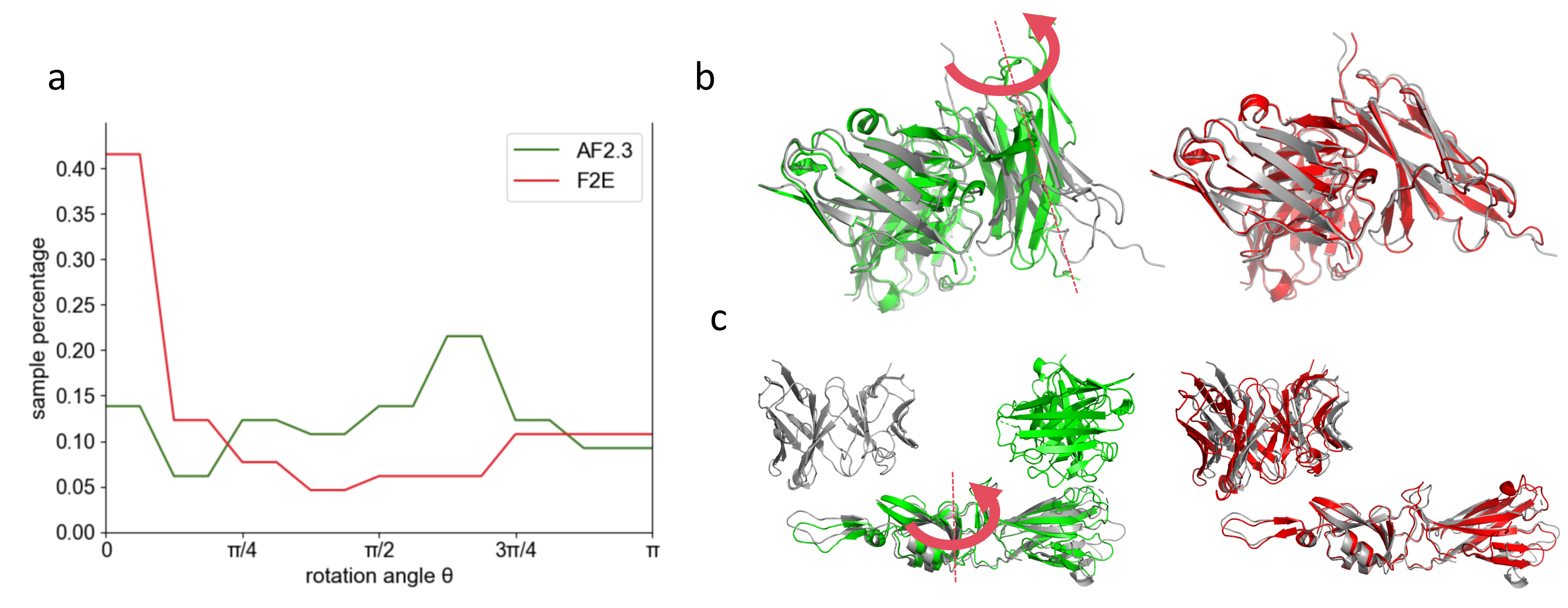}}
\vskip -0.1in
\caption{
% Case study and error analysis. RW: remove for space saving
(\textbf{a}) The visualization of rotational errors for AF2.3 and F2E. (\textbf{b}) The predicted structures of PDB 7Y1B against ground truth. The grey structure is the experimental ground truth, the green structure by AF2.3 and the red structure by model tuned with F2E. We show in red arrows how to rotate from the wrong prediction to the ground truth. (\textbf{c}) The predicted structures of 8GPT against ground truth.
}
\label{fig:error}
\end{center}
\vskip -0.3in
\end{figure*}

In \Cref{fig:error} (\textbf{a}) we show the change of rotation error distributions before and after fine-tuning with F2E.
We can see that the rotation error is dramatically reduced with many data points drop below $\frac{\pi}{4}$.
In \Cref{fig:error} (\textbf{b}) and (\textbf{c}) we show the predictions results of AF2.3 and F2E on 7Y1B and 8GPT.
By correcting the rotation error, F2E helps the model the make correct predictions.
Notice that in (\textbf{b}) the alignment is conducted on antibody structures, and no translational error can be observed in this way.
In (\textbf{c}) the alignment is conducted on antigen structures, and the pose error not only introduces rotational error but also translational error.
Nevertheless, such kind of error can be corrected by only rotating the antigen alone the axis shown in the left image.
A meaningful discussion about calculating relative pose error from different perspectives is provided in \Cref{appendixsub:delta_delta_frame}.

\section{Related work} \label{main:related_work}
\textbf{Protein structure prediction} The success of AF2~\cite{jumper2021af2} has brought an increasing interest in developing new protein structure prediction models. Some of these efforts~\cite{baek2021rf, wu2022omegafold, lin2022esmfold, baek2023rf2} try to improve the inference speed of AF2 either by speeding up the neural network or reducing the demand for MSA searching. \citet{evans2021af2multimer}, \citet{gao2022af2complex} and \citet{baek2023rf2} are later developed to predict complex structures. Recently, several new models~\cite{baek2023rf2na,krishna2023rf2full_atom,2023af2latest} further extend their ability to small molecules, nucleic acids and other hetero components. Several papers~\cite{wallner2023af2sample, bryant2023af2denoise, yin2023af2evaluation, gaudreault2023enhanced} propose to apply or enhance AF2 on immune complex modeling. However, those methods mainly focus on inference-time enhancement.

\textbf{Protein rigid docking}
Many algorithms are designed to solve the problem of rigid docking problem from two known structures. There have been many efforts~\cite{pierce2014zdock, yan2020hdock, desta2020cluspro} taken with grid-search-based or template-based approaches. Recently, deep learning based docking~\cite{ganea2021independent,jin2022antibody,wang2023injecting,ketata2023diffdock} gains more and more attention. Our setting is different from rigid docking with no known structures as input.

\textbf{Parameter efficient fine-tuning} We use LoRA~\citep{hu2022lora} in our work as the parameter-efficient finetuning method selected among a wealth of recently proposed methods~\citep{lester-etal-2021-power, he2022towards, dettmers2023qlora, xia2024chain} for its simplicity and speed. 
These methods reduce the cost of fine-tuning large models by training either additional or a subset of parameters and often reach comparable performance with full-parameter finetuning.

\section{Conclusion} \label{main:conclusion}
Computational immune complex modeling is a challenging task which is important to downstream applications in antibody drug design. In this paper, we propose F2E, a geodesic distance loss that can solve the gradient vanishing problem in the original FAPE loss from AF2~\cite{gao2022af2complex}. 
% By derivations, 
We introduce a theoretically sound way to correct inter-chain FAPE loss as Group-aware FAPE loss, which becomes a chordal distance loss for relative pose errors of nearly-rigid groups defined on protein chains. The similar derivation can help us find a way to correct inter-chain F2E as Group-aware F2E loss. We conduct extensive experiments and analysis to show the effectiveness of our approach.

For future directions, we expect to extend our loss to different complex modeling problems, including those with non-protein components. Indeed, the implicit definition of group frame indicates that F2E can be easily adapted to these components regardless of their inner geometry. Even though our experiments are conducted by fine-tuning AF2, we expect F2E can be also beneficial during pre-training stage of protein foundation models. We really expect to apply F2E to multiple scenarios to test its generalizability.

\section*{Acknowledgements}
This work was supported by the National Key Research and Development Program of China grants 2022YFF1203100.

\section*{Impact Statement} \label{main:broader_impact}
This work aims to advance machine learning models in protein complex modeling.
This is an important issue in many aspects of the pharmaceutical industry. 
For instance, antibody antigen complex modeling is of great importance in developing vaccines.
High performance models in this area have the potential of revolutionizing the field of drug discovery.

\bibliography{fafe}

\begin{thebibliography}{46}
\providecommand{\natexlab}[1]{#1}
\providecommand{\url}[1]{\texttt{#1}}
\expandafter\ifx\csname urlstyle\endcsname\relax
  \providecommand{\doi}[1]{doi: #1}\else
  \providecommand{\doi}{doi: \begingroup \urlstyle{rm}\Url}\fi

\bibitem[Baek et~al.(2021)Baek, DiMaio, Anishchenko, Dauparas, Ovchinnikov, Lee, Wang, Cong, Kinch, Schaeffer, et~al.]{baek2021rf}
Baek, M., DiMaio, F., Anishchenko, I., Dauparas, J., Ovchinnikov, S., Lee, G.~R., Wang, J., Cong, Q., Kinch, L.~N., Schaeffer, R.~D., et~al.
\newblock Accurate prediction of protein structures and interactions using a three-track neural network.
\newblock \emph{Science}, 373\penalty0 (6557):\penalty0 871--876, 2021.

\bibitem[Baek et~al.(2023{\natexlab{a}})Baek, Anishchenko, Humphreys, Cong, Baker, and DiMaio]{baek2023rf2}
Baek, M., Anishchenko, I., Humphreys, I., Cong, Q., Baker, D., and DiMaio, F.
\newblock Efficient and accurate prediction of protein structure using rosettafold2.
\newblock \emph{bioRxiv}, pp.\  2023--05, 2023{\natexlab{a}}.

\bibitem[Baek et~al.(2023{\natexlab{b}})Baek, McHugh, Anishchenko, Jiang, Baker, and DiMaio]{baek2023rf2na}
Baek, M., McHugh, R., Anishchenko, I., Jiang, H., Baker, D., and DiMaio, F.
\newblock Accurate prediction of protein--nucleic acid complexes using rosettafoldna.
\newblock \emph{Nature Methods}, pp.\  1--5, 2023{\natexlab{b}}.

\bibitem[Bai et~al.(2015)Bai, McMullan, and Scheres]{bai2015experiment2}
Bai, X.-C., McMullan, G., and Scheres, S.~H.
\newblock How cryo-em is revolutionizing structural biology.
\newblock \emph{Trends in biochemical sciences}, 40\penalty0 (1):\penalty0 49--57, 2015.

\bibitem[Basu \& Wallner(2016)Basu and Wallner]{basu2016dockq}
Basu, S. and Wallner, B.
\newblock Dockq: a quality measure for protein-protein docking models.
\newblock \emph{PloS one}, 11\penalty0 (8):\penalty0 e0161879, 2016.

\bibitem[Braslavsky et~al.(2003)Braslavsky, Hebert, Kartalov, and Quake]{braslavsky2003sequence}
Braslavsky, I., Hebert, B., Kartalov, E., and Quake, S.~R.
\newblock Sequence information can be obtained from single dna molecules.
\newblock \emph{Proceedings of the National Academy of Sciences}, 100\penalty0 (7):\penalty0 3960--3964, 2003.

\bibitem[Bryant \& Noe(2023)Bryant and Noe]{bryant2023af2denoise}
Bryant, P. and Noe, F.
\newblock Improved protein complex prediction with alphafold-multimer by denoising the msa profile.
\newblock \emph{bioRxiv}, pp.\  2023--07, 2023.

\bibitem[Burley et~al.(2017)Burley, Berman, Kleywegt, Markley, Nakamura, and Velankar]{burley2017pdb}
Burley, S.~K., Berman, H.~M., Kleywegt, G.~J., Markley, J.~L., Nakamura, H., and Velankar, S.
\newblock Protein data bank (pdb): the single global macromolecular structure archive.
\newblock \emph{Protein crystallography: methods and protocols}, pp.\  627--641, 2017.

\bibitem[Carlone et~al.(2022)Carlone, Khosoussi, Tzoumas, Habibi, Ryll, Talak, Shi, and Antonante]{carlone2022visual}
Carlone, L., Khosoussi, K., Tzoumas, V., Habibi, G., Ryll, M., Talak, R., Shi, J., and Antonante, P.
\newblock Visual navigation for autonomous vehicles: An open-source hands-on robotics course at mit.
\newblock In \emph{2022 IEEE Integrated STEM Education Conference (ISEC)}, pp.\  177--184. IEEE, 2022.

\bibitem[Consortium(2019)]{uniprot2019uniprot}
Consortium, U.
\newblock Uniprot: a worldwide hub of protein knowledge.
\newblock \emph{Nucleic acids research}, 47\penalty0 (D1):\penalty0 D506--D515, 2019.

\bibitem[DeepMind(2023)]{2023af2latest}
DeepMind.
\newblock A glimpse of the next generation of alphafold, 2023.

\bibitem[Desta et~al.(2020)Desta, Porter, Xia, Kozakov, and Vajda]{desta2020cluspro}
Desta, I.~T., Porter, K.~A., Xia, B., Kozakov, D., and Vajda, S.
\newblock Performance and its limits in rigid body protein-protein docking.
\newblock \emph{Structure}, 28\penalty0 (9):\penalty0 1071--1081, 2020.

\bibitem[Dettmers et~al.(2023)Dettmers, Pagnoni, Holtzman, and Zettlemoyer]{dettmers2023qlora}
Dettmers, T., Pagnoni, A., Holtzman, A., and Zettlemoyer, L.
\newblock Qlora: Efficient finetuning of quantized llms.
\newblock \emph{arXiv preprint arXiv:2305.14314}, 2023.

\bibitem[Dunbar \& Deane(2016)Dunbar and Deane]{dunbar2016anarci}
Dunbar, J. and Deane, C.~M.
\newblock Anarci: antigen receptor numbering and receptor classification.
\newblock \emph{Bioinformatics}, 32\penalty0 (2):\penalty0 298--300, 2016.

\bibitem[Dunbar et~al.(2014)Dunbar, Krawczyk, Leem, Baker, Fuchs, Georges, Shi, and Deane]{dunbar2014sabdab}
Dunbar, J., Krawczyk, K., Leem, J., Baker, T., Fuchs, A., Georges, G., Shi, J., and Deane, C.~M.
\newblock Sabdab: the structural antibody database.
\newblock \emph{Nucleic acids research}, 42\penalty0 (D1):\penalty0 D1140--D1146, 2014.

\bibitem[Evans et~al.(2021)Evans, O’Neill, Pritzel, Antropova, Senior, Green, {\v{Z}}{\'\i}dek, Bates, Blackwell, Yim, et~al.]{evans2021af2multimer}
Evans, R., O’Neill, M., Pritzel, A., Antropova, N., Senior, A., Green, T., {\v{Z}}{\'\i}dek, A., Bates, R., Blackwell, S., Yim, J., et~al.
\newblock Protein complex prediction with alphafold-multimer.
\newblock \emph{biorxiv}, pp.\  2021--10, 2021.

\bibitem[Ganea et~al.(2021)Ganea, Huang, Bunne, Bian, Barzilay, Jaakkola, and Krause]{ganea2021independent}
Ganea, O.-E., Huang, X., Bunne, C., Bian, Y., Barzilay, R., Jaakkola, T., and Krause, A.
\newblock Independent se (3)-equivariant models for end-to-end rigid protein docking.
\newblock \emph{arXiv preprint arXiv:2111.07786}, 2021.

\bibitem[Gao et~al.(2022)Gao, Nakajima~An, Parks, and Skolnick]{gao2022af2complex}
Gao, M., Nakajima~An, D., Parks, J.~M., and Skolnick, J.
\newblock Af2complex predicts direct physical interactions in multimeric proteins with deep learning.
\newblock \emph{Nature communications}, 13\penalty0 (1):\penalty0 1744, 2022.

\bibitem[Gaudreault et~al.(2023)Gaudreault, Corbeil, and Sulea]{gaudreault2023enhanced}
Gaudreault, F., Corbeil, C.~R., and Sulea, T.
\newblock Enhanced antibody-antigen structure prediction from molecular docking using alphafold2.
\newblock \emph{Scientific Reports}, 13\penalty0 (1):\penalty0 15107, 2023.

\bibitem[Harris et~al.(2008)Harris, Buzby, Babcock, Beer, Bowers, Braslavsky, Causey, Colonell, DiMeo, Efcavitch, et~al.]{harris2008single}
Harris, T.~D., Buzby, P.~R., Babcock, H., Beer, E., Bowers, J., Braslavsky, I., Causey, M., Colonell, J., DiMeo, J., Efcavitch, J.~W., et~al.
\newblock Single-molecule dna sequencing of a viral genome.
\newblock \emph{Science}, 320\penalty0 (5872):\penalty0 106--109, 2008.

\bibitem[Hartley et~al.(2013)Hartley, Trumpf, Dai, and Li]{hartley2013rotation}
Hartley, R., Trumpf, J., Dai, Y., and Li, H.
\newblock Rotation averaging.
\newblock \emph{International journal of computer vision}, 103:\penalty0 267--305, 2013.

\bibitem[He et~al.(2022)He, Zhou, Ma, Berg-Kirkpatrick, and Neubig]{he2022towards}
He, J., Zhou, C., Ma, X., Berg-Kirkpatrick, T., and Neubig, G.
\newblock Towards a unified view of parameter-efficient transfer learning.
\newblock In \emph{International Conference on Learning Representations}, 2022.
\newblock URL \url{https://openreview.net/forum?id=0RDcd5Axok}.

\bibitem[Heather \& Chain(2016)Heather and Chain]{heather2016sequence}
Heather, J.~M. and Chain, B.
\newblock The sequence of sequencers: The history of sequencing dna.
\newblock \emph{Genomics}, 107\penalty0 (1):\penalty0 1--8, 2016.

\bibitem[Hu et~al.(2022)Hu, Shen, Wallis, Allen-Zhu, Li, Wang, Wang, and Chen]{hu2022lora}
Hu, E.~J., Shen, Y., Wallis, P., Allen-Zhu, Z., Li, Y., Wang, S., Wang, L., and Chen, W.
\newblock Lo{RA}: Low-rank adaptation of large language models.
\newblock In \emph{International Conference on Learning Representations}, 2022.
\newblock URL \url{https://openreview.net/forum?id=nZeVKeeFYf9}.

\bibitem[Huynh(2009)]{huynh2009metrics}
Huynh, D.~Q.
\newblock Metrics for 3d rotations: Comparison and analysis.
\newblock \emph{Journal of Mathematical Imaging and Vision}, 35:\penalty0 155--164, 2009.

\bibitem[Jaskolski et~al.(2014)Jaskolski, Dauter, and Wlodawer]{jaskolski2014experiment3}
Jaskolski, M., Dauter, Z., and Wlodawer, A.
\newblock A brief history of macromolecular crystallography, illustrated by a family tree and its n obel fruits.
\newblock \emph{The FEBS journal}, 281\penalty0 (18):\penalty0 3985--4009, 2014.

\bibitem[Jin et~al.(2022)Jin, Barzilay, and Jaakkola]{jin2022antibody}
Jin, W., Barzilay, R., and Jaakkola, T.
\newblock Antibody-antigen docking and design via hierarchical structure refinement.
\newblock In \emph{International Conference on Machine Learning}, pp.\  10217--10227. PMLR, 2022.

\bibitem[Jumper et~al.(2021)Jumper, Evans, Pritzel, Green, Figurnov, Ronneberger, Tunyasuvunakool, Bates, {\v{Z}}{\'\i}dek, Potapenko, et~al.]{jumper2021af2}
Jumper, J., Evans, R., Pritzel, A., Green, T., Figurnov, M., Ronneberger, O., Tunyasuvunakool, K., Bates, R., {\v{Z}}{\'\i}dek, A., Potapenko, A., et~al.
\newblock Highly accurate protein structure prediction with alphafold.
\newblock \emph{Nature}, 596\penalty0 (7873):\penalty0 583--589, 2021.

\bibitem[Kaczor et~al.(2018)Kaczor, Bartuzi, Stepniewski, Matosiuk, and Selent]{kaczor2018ppd}
Kaczor, A.~A., Bartuzi, D., Stepniewski, T.~M., Matosiuk, D., and Selent, J.
\newblock Protein--protein docking in drug design and discovery.
\newblock \emph{Computational Drug Discovery and Design}, pp.\  285--305, 2018.

\bibitem[Ketata et~al.(2023)Ketata, Laue, Mammadov, St{\"a}rk, Wu, Corso, Marquet, Barzilay, and Jaakkola]{ketata2023diffdock}
Ketata, M.~A., Laue, C., Mammadov, R., St{\"a}rk, H., Wu, M., Corso, G., Marquet, C., Barzilay, R., and Jaakkola, T.~S.
\newblock Diffdock-pp: Rigid protein-protein docking with diffusion models.
\newblock \emph{arXiv preprint arXiv:2304.03889}, 2023.

\bibitem[Krishna et~al.(2023)Krishna, Wang, Ahern, Sturmfels, Venkatesh, Kalvet, Lee, Morey-Burrows, Anishchenko, Humphreys, et~al.]{krishna2023rf2full_atom}
Krishna, R., Wang, J., Ahern, W., Sturmfels, P., Venkatesh, P., Kalvet, I., Lee, G.~R., Morey-Burrows, F.~S., Anishchenko, I., Humphreys, I.~R., et~al.
\newblock Generalized biomolecular modeling and design with rosettafold all-atom.
\newblock \emph{bioRxiv}, pp.\  2023--10, 2023.

\bibitem[Kryshtafovych et~al.(2021)Kryshtafovych, Schwede, Topf, Fidelis, and Moult]{kryshtafovych2021casp14}
Kryshtafovych, A., Schwede, T., Topf, M., Fidelis, K., and Moult, J.
\newblock Critical assessment of methods of protein structure prediction (casp)—round xiv.
\newblock \emph{Proteins: Structure, Function, and Bioinformatics}, 89\penalty0 (12):\penalty0 1607--1617, 2021.

\bibitem[Lester et~al.(2021)Lester, Al-Rfou, and Constant]{lester-etal-2021-power}
Lester, B., Al-Rfou, R., and Constant, N.
\newblock The power of scale for parameter-efficient prompt tuning.
\newblock In Moens, M.-F., Huang, X., Specia, L., and Yih, S. W.-t. (eds.), \emph{Proceedings of the 2021 Conference on Empirical Methods in Natural Language Processing}, pp.\  3045--3059, Online and Punta Cana, Dominican Republic, November 2021. Association for Computational Linguistics.
\newblock \doi{10.18653/v1/2021.emnlp-main.243}.
\newblock URL \url{https://aclanthology.org/2021.emnlp-main.243}.

\bibitem[Lin et~al.(2022)Lin, Akin, Rao, Hie, Zhu, Lu, dos Santos~Costa, Fazel-Zarandi, Sercu, Candido, et~al.]{lin2022esmfold}
Lin, Z., Akin, H., Rao, R., Hie, B., Zhu, Z., Lu, W., dos Santos~Costa, A., Fazel-Zarandi, M., Sercu, T., Candido, S., et~al.
\newblock Language models of protein sequences at the scale of evolution enable accurate structure prediction.
\newblock \emph{BioRxiv}, 2022:\penalty0 500902, 2022.

\bibitem[M{\'e}ndez et~al.(2003)M{\'e}ndez, Leplae, De~Maria, and Wodak]{mendez2003ligandrmsd}
M{\'e}ndez, R., Leplae, R., De~Maria, L., and Wodak, S.~J.
\newblock Assessment of blind predictions of protein--protein interactions: current status of docking methods.
\newblock \emph{Proteins: Structure, Function, and Bioinformatics}, 52\penalty0 (1):\penalty0 51--67, 2003.

\bibitem[Pierce et~al.(2014)Pierce, Wiehe, Hwang, Kim, Vreven, and Weng]{pierce2014zdock}
Pierce, B.~G., Wiehe, K., Hwang, H., Kim, B.-H., Vreven, T., and Weng, Z.
\newblock Zdock server: interactive docking prediction of protein--protein complexes and symmetric multimers.
\newblock \emph{Bioinformatics}, 30\penalty0 (12):\penalty0 1771--1773, 2014.

\bibitem[Richardson et~al.(2023)Richardson, Allen, Baldi, Beracochea, Bileschi, Burdett, Burgin, Caballero-P{\'e}rez, Cochrane, Colwell, et~al.]{richardson2023mgnify}
Richardson, L., Allen, B., Baldi, G., Beracochea, M., Bileschi, M.~L., Burdett, T., Burgin, J., Caballero-P{\'e}rez, J., Cochrane, G., Colwell, L.~J., et~al.
\newblock Mgnify: the microbiome sequence data analysis resource in 2023.
\newblock \emph{Nucleic Acids Research}, 51\penalty0 (D1):\penalty0 D753--D759, 2023.

\bibitem[Steinegger \& S{\"o}ding(2017)Steinegger and S{\"o}ding]{steinegger2017mmseqs2}
Steinegger, M. and S{\"o}ding, J.
\newblock Mmseqs2 enables sensitive protein sequence searching for the analysis of massive data sets.
\newblock \emph{Nature biotechnology}, 35\penalty0 (11):\penalty0 1026--1028, 2017.

\bibitem[Thompson et~al.(2020)Thompson, Yeates, and Rodriguez]{thompson2020experiment1}
Thompson, M.~C., Yeates, T.~O., and Rodriguez, J.~A.
\newblock Advances in methods for atomic resolution macromolecular structure determination.
\newblock \emph{F1000Research}, 9, 2020.

\bibitem[Wallner(2023)]{wallner2023af2sample}
Wallner, B.
\newblock Afsample: improving multimer prediction with alphafold using massive sampling.
\newblock \emph{Bioinformatics}, 39\penalty0 (9):\penalty0 btad573, 2023.

\bibitem[Wang et~al.(2023)Wang, Sun, Luo, Li, Yang, Cheng, Li, Shi, and Song]{wang2023injecting}
Wang, R., Sun, Y., Luo, Y., Li, S., Yang, C., Cheng, X., Li, H., Shi, C., and Song, L.
\newblock Injecting multimodal information into rigid protein docking via bi-level optimization.
\newblock In \emph{Thirty-seventh Conference on Neural Information Processing Systems}, 2023.

\bibitem[Wu et~al.(2022)Wu, Ding, Wang, Shen, Zhang, Luo, Su, Wu, Xie, Berger, et~al.]{wu2022omegafold}
Wu, R., Ding, F., Wang, R., Shen, R., Zhang, X., Luo, S., Su, C., Wu, Z., Xie, Q., Berger, B., et~al.
\newblock High-resolution de novo structure prediction from primary sequence.
\newblock \emph{BioRxiv}, pp.\  2022--07, 2022.

\bibitem[W{\"u}thrich(2001)]{wuthrich2001experiment4}
W{\"u}thrich, K.
\newblock The way to nmr structures of proteins.
\newblock \emph{Nature structural biology}, 8\penalty0 (11):\penalty0 923--925, 2001.

\bibitem[Xia et~al.(2024)Xia, Qin, and Hazan]{xia2024chain}
Xia, W., Qin, C., and Hazan, E.
\newblock Chain of lora: Efficient fine-tuning of language models via residual learning, 2024.

\bibitem[Yan et~al.(2020)Yan, Tao, He, and Huang]{yan2020hdock}
Yan, Y., Tao, H., He, J., and Huang, S.-Y.
\newblock The hdock server for integrated protein--protein docking.
\newblock \emph{Nature protocols}, 15\penalty0 (5):\penalty0 1829--1852, 2020.

\bibitem[Yin \& Pierce(2023)Yin and Pierce]{yin2023af2evaluation}
Yin, R. and Pierce, B.~G.
\newblock Evaluation of alphafold antibody-antigen modeling with implications for improving predictive accuracy.
\newblock \emph{Protein Science}, pp.\  e4865, 2023.

\end{thebibliography}
\bibliographystyle{icml2024}

%%%%%%%%%%%%%%%%%%%%%%%%%%%%%%%%%%%%%%%%%%%%%%%%%%%%%%%%%%%%%%%%%%%%%%%%%%%%%%%
%%%%%%%%%%%%%%%%%%%%%%%%%%%%%%%%%%%%%%%%%%%%%%%%%%%%%%%%%%%%%%%%%%%%%%%%%%%%%%%
% APPENDIX
%%%%%%%%%%%%%%%%%%%%%%%%%%%%%%%%%%%%%%%%%%%%%%%%%%%%%%%%%%%%%%%%%%%%%%%%%%%%%%%
%%%%%%%%%%%%%%%%%%%%%%%%%%%%%%%%%%%%%%%%%%%%%%%%%%%%%%%%%%%%%%%%%%%%%%%%%%%%%%%
\newpage
\appendix
\onecolumn
\section{SO(3) and SE(3) distance metrics} \label{appendix:metric}
\subsection{The distance of two frames can be calculated by their composition}
\label{appendix:sub:def_compose}
\Cref{def:delta_frame} relies on the proof that the distance of two frames can be calculated by their composition. Here we provide a brief proof for it.
\begin{proof}
  According to \Cref{def:chordal_se3} and \Cref{def:geodesic_se3}, we have \\
  $dist({\hat{T}}, T) = f(\hat{R}^\top R, \vec{t}-\hat{\vec{t}}),$ \\
  in which $dist(\cdot)$ denotes either chordal or geodesic distance. \\
  And we know that the composition of $\hat{T}^{-1} T$ equals to $(\hat{R}^\top R, \hat{R}^\top(\vec{t}-\hat{\vec{t}}))$, \\
  since $f(\hat{R}R^\top, t-\hat{t}) = f(\hat{R}R^\top, \hat{R}^\top(\vec{t}-\hat{\vec{t}}))$, \\
  we have $dist({\hat{T}}, T) = f(\hat{T}^{-1} T)$. \\
  So we have $dist({\hat{T}}^{-1}, T^{-1}) = f(\hat{T} T^{-1})$.
\end{proof}
\subsection{Relationship of $\SO3$ chordal and geodesic distance}
\begin{proposition}
    Given two rotations $R_i$ and $R_j \in \SO3$, we have
    \begin{equation}
        d_{c}(R_i, R_j) = 2\sqrt{2}sin(\frac{d_{\theta}(R_i, R_j)}{2}),
    \end{equation}
    in which $d_c(*)$ denotes the chordal distance defined in \Cref{def:chordal_so3} and $d_{\theta}(*)$ denotes the geodesic distance defined in \Cref{def:geodesic_so3}.
\end{proposition}
\begin{proof}
\label{proof:geo2chordal}
  \begin{equation}
    \begin{aligned}
      d_{c}(R_i, R_j) & = {\Vert R_i - R_j \Vert}_F \\
                      & = {\Vert R_i^\top R_j - I \Vert}_F \\
                      & = 2(\sin^2(d_{\theta}(R_i, R_j)) +  (1 - \cos(d_{\theta}(R_i, R_j)))^2) \quad\quad \mathrm{(Rodrigues’\; formula)}\\
                      & = 2\sqrt{2} \sin(\frac{d_{\theta}(R_i, R_j)}{2})
    \end{aligned}
  \end{equation}
\end{proof}

\subsection{Validation of metrics}
\label{appendixsub:valid_metrics}
Many excellent reviews \cite{huynh2009metrics, hartley2013rotation, carlone2022visual} have discussed the distance metrics on $\SO3$ and $\SE3$. Here we only provide brief proof that the metrics introduced in \ref{mainsub:so3_se3_metrics} are valid. A \emph{valid} distance $dist(a, b)$ between two generic elements “a” and “b” must satisfy the following properties.
\begin{align}
  dist(a, b) & \geq 0 & \mathrm{(non-negativity)} \\
  dist(a, b) & = 0 \iff a = b & \mathrm{(identity)} \\
  dist(a, b) & = dist(b, a) & \mathrm{(symmetry)} \\
  dist(a, c) & \leq dist(a, b) + dist(b, c) & \mathrm{(triangle \ inequality)}
\end{align}
We prove the chordal distance on $\SO3$ defined in \Cref{def:chordal_so3} satisfies the properties above.
\begin{proposition}
  Given two rotations $R_i$ and $R_j \in \SO3$, their chordal distance defined as \\
  \begin{equation}
    dist_{c}(R_i, R_j) = {\Vert R_i - R_j\Vert}_F = {\Vert R_i^\top R_j - I\Vert}_F
  \end{equation}
  is a valid distance metric.
\end{proposition}
\begin{proof}
The proof for the first three properties is straightforward. Thus we only prove the triangle inequality.
  \begin{equation}
      \begin{aligned}
                     dist_c(R_a, R_c) & \leq dist_c(R_a, R_b) + dist_c(R_b, R_c) \\
      \iff {\Vert R_a - R_c\Vert}_F & \leq {\Vert R_a - R_b\Vert}_F + {\Vert R_b - R_c\Vert}_F \\
      \iff {\Vert \mathop{vec}(R_a) - \mathop{vec}(R_c)\Vert}_F & \leq {\Vert \mathop{vec}(R_a) - \mathop{vec}(R_b)\Vert}_F + {\Vert \mathop{vec}(R_b) - \mathop{vec}(R_c)\Vert}_F, \\
      \end{aligned}
  \end{equation}
  which can be implied from triangle inequation for vectors.
\end{proof}
The validation of geodesic distance on $\SO3$ is by definition, as the geodesic distance equals the shortest path between two points on the manifold. The validation of metrics on $\SE3$ defined in \Cref{def:chordal_se3} and \Cref{def:geodesic_se3} are also trivial, as they treat $\SE3$ as the Cartesian product of $\SO3$ and $\R^3$.

\subsection{Relative pose error}
\label{appendixsub:delta_delta_frame}
Chordal distance and geodesic distance on $\SE3$ are left-invariant only. This means usually we have
\begin{align}
\label{eq:different1}
    dist({\hat{T}_A}^{-1} \hat{T}_B, {T_A}^{-1} T_B) & \neq dist({\hat{T}_A}^{-1} T_A, {\hat{T}_B}^{-1} T_B) \\
\label{eq:different2}
    dist({\hat{T}_A}^{-1} \hat{T}_B, {T_A}^{-1} T_B) & \neq dist({\hat{T}_B}^{-1} \hat{T}_A, {T_B}^{-1} T_A),
\end{align}
which means there can be multiple inequivalent approaches to optimize relative pose error depending on the definition. However, the good news is that chordal distance and geodesic distance on $\SO3$ are left-invariant and right-invariant, so their rotation errors are the same. Their optimization destination is the same too, because once the rotation error is reduced to zero, their translation errors become equivalent.
We prove that the rotation errors on both sides in Inequation \ref{eq:different1} are equivalent.
\begin{proposition}
  \begin{equation}
    dist_c({\hat{R}_i}^{\top} R_i, {\hat{R}_j}^{\top} R_j) = dist_c({\hat{R}_i}^{\top} \hat{R}_j, {R_i}^{\top} R_j)
  \end{equation}
\end{proposition}
\begin{proof}
  \begin{equation}
      \begin{aligned}
        & \; dist_c({\hat{R}_i}^{\top} R_i, {\hat{R}_j}^{\top} R_j) \\
      = & \; {{\Vert \hat{R}_i}^{\top} R_i - {\hat{R}_j}^{\top} R_j \Vert}_F \\
      = & \; {{\Vert \hat{R}_j \hat{R}_i}^{\top} R_i - R_j \Vert}_F \\
      = & \; {{\Vert \hat{R}_j \hat{R}_i}^{\top}  - R_j {R_i}^{\top} \Vert}_F \\
      = & \; {\Vert {\hat{R}_i}^{\top} \hat{R}_j  - {R_i}^{\top} R_j \Vert}_F \\
      = & \; dist_c({\hat{R}_i}^{\top} \hat{R}_j, {R_i}^{\top} R_j)
      \end{aligned}
  \end{equation}
\end{proof}
\begin{proposition}
  \begin{equation}
    dist_\theta({\hat{R}_i}^{\top} R_i, {\hat{R}_j}^{\top} R_j) = dist_\theta({\hat{R}_i}^{\top} \hat{R}_j, {R_i}^{\top} R_j)
  \end{equation}
\end{proposition}
\begin{proof}
  \begin{equation}
      \begin{aligned}
          & \; dist_\theta({\hat{R}_i}^{\top} R_i, {\hat{R}_j}^{\top} R_j) \\
        = & \; |\arccos(\frac{\Tr({({\hat{R}_i}^{\top} R_i)}^{\top} {\hat{R}_j}^{\top} R_j)-1}{2}) | \\
        = & \; |\arccos(\frac{\Tr({({\hat{R}_i}^{\top} {\hat{R}_j})}^{\top} {R_i}^{\top} R_j)-1}{2}) | \\
        = & \; dist_\theta({\hat{R}_i}^{\top} \hat{R}_j, {R_i}^{\top} R_j)
      \end{aligned}
  \end{equation}
\end{proof}
The rotation error of Inequation \ref{eq:different2} can be proved in the same way.

\section{FAPE loss and F2E loss details} \label{appendix:loss}
\subsection{FAPE loss details}
In the original AF2 \cite{jumper2021af2} paper, There are two places where FAPE loss is applied. One is called auxiliary FAPE loss which is calculated in the intermediate backbone prediction of each cycle in the structure module. The other is called full FAPE loss which is calculated in the final full-atom prediction from the last cycle, which is reconstructed from backbone frames and sidechain torsion angles predicted. We denote the FAPE loss on different cycles as
\begin{definition}
\label{def:fape_each_cycle}
  Given any residue frame $i$ and any atom $j$ on the protein structure, we have $\hat{T}_i, T_i\in \mathrm{SE}(3)$ and $\hat{\vec{x}}_j, \vec{x}_j\in \mathbb{R}^3$, the FAPE loss on the predicted structure on cycle $c\in [0, N_{cycle}-1]$ is defined as
  \begin{equation}
  L_{\mathrm{FAPE}}(i, j, c) = {\Vert  {\hat{T}_{i,c}}^{-1} \circ \hat{\vec{x}}_{j,c} - {T_{i,c}}^{-1} \circ \vec{x}_{j,c} \Vert}_F,
  \end{equation}
  in which ${\Vert * \Vert}_F$ denotes the Frobenius norm.
\end{definition}
The FAPE loss applied to monomers can then be written as
\begin{align} \label{eq:monomer_fape}
  & L = L_{\text{Full FAPE}} + L_{\text{Auxiliary FAPE}} \\
  L_{\text{Full FAPE}} & = \frac{1}{N_{pair}} \sum\limits_{i \in \{\text{full frames}\}, j \in \{\text{full atoms}\}}L_{\text{FAPE}}(i, j, c=N-1) \\
  L_{\text{Auxiliary FAPE}} & = \frac{1}{N_{cycle} N_{pair}} \sum\limits_{c=0}^{N-1}\sum\limits_{i \in \{\text{backbone frames}\}, j \in \{\text{C-alpha atoms}\}}L_{\text{FAPE}}(i, j, c),
\end{align}
in which $N_{pair}$ is determined per summation notation. The FAPE loss applied to multimers can be written as
\begin{align} \label{eq:multimer_fape}
  & L = L_{\text{Full FAPE}} + L_{\text{Intra-chain auxiliary FAPE}} + L_{\text{Inter-chain auxiliary FAPE}} \\
  L_{\text{Full FAPE}} & = \frac{1}{N_{pair}} \sum\limits_{i \in \{\text{full frames}\}, j \in \{\text{full atoms}\}}L_{\text{FAPE}}(i, j, c=N-1) \\
  L_{\text{Intra-chain auxiliary FAPE}} & = \frac{1}{N_{cycle} N_{pair}} \sum\limits_{c=0}^{N-1}\sum\limits_{
    i \in \{\text{backbone frames}\}, j \in \{\text{C-alpha atoms}\},\text{chain(i)}= \text{chain(j)}}L_{\text{FAPE}}(i, j, c) \\
  L_{\text{Inter-chain auxiliary FAPE}} & = \frac{1}{N_{cycle} N_{pair}} \sum\limits_{c=0}^{N-1}\sum\limits_{
    i \in \{\text{backbone frames}\}, j \in \{\text{C-alpha atoms}\},\text{chain(i)}\neq \text{chain(j)}}L_{\text{FAPE}}(i, j, c) \\
\end{align}
When using the group-aware version of FAPE, we only change the inter-chain part as
\begin{equation}
  L_{\text{Inter-chain auxiliary FAPE, group-aware}} = \frac{1}{N_{cycle} N_{pair}} \sum\limits_{c=0}^{N-1}\sqrt{\sum\limits_{
    i \in \{\text{backbone frames}\}, j \in \{\text{C-alpha atoms}\},\text{chain(i)}\neq \text{chain(j)}}L_{\text{FAPE}}(i, j, c)^2}. \\
\end{equation}
Remind that from the derivations in \cref{mainsub:g-fape}, only the quadratic mean over $j$ is required, But for implementation simplicity, we apply it to both $i$ and $j$. This works well empirically.

\subsection{F2E loss details}
We denote the F2E loss on different cycles as
\begin{definition}
\label{def:f2e_each_cycle}
  Given any residue frame $i$ and any frame $j$ on the protein structure, we have $\hat{T}_i, T_i, \hat{T}_j, T_j\in \mathrm{SE}(3)$, the F2E loss on the predicted structure on cycle $c\in [0, N-1]$ is defined as
  \begin{equation}
    L_{\mathrm{F2E}}(i, j, c; \alpha) = dist_{\theta}({\hat{T}_{i, c}}^{-1} \hat{T}_{j, c}, {T_{i, c}}^{-1} T_{j, c}; \alpha),
  \end{equation}
  in which $dist_{\theta}(*)$ denotes the geodesic distance defined in \Cref{def:geodesic_se3}.
\end{definition}

When we instead use the loss of F2E, we apply F2E to all terms in \ref{eq:multimer_fape}. So we get
\begin{align} \label{eq:multimer_f2e}
  & L = L_{\text{Full F2E}} + L_{\text{Intra-chain auxiliary F2E}} + L_{\text{Inter-chain auxiliary F2E}} \\
  L_{\text{Full F2E}} & = \frac{1}{N_{pair}} \sum\limits_{i \in \{\text{full frames}\}, j \in \{\text{full frames}\}}L_{\text{F2E}}(i, j, c=N-1) \\
  L_{\text{Intra-chain auxiliary F2E}} & = \frac{1}{N_{cycle} N_{pair}} \sum\limits_{c=0}^{N-1}\sum\limits_{
    i \in \{\text{backbone frames}\}, j \in \{\text{backbone frames}\},\text{chain(i)}= \text{chain(j)}}L_{\text{F2E}}(i, j, c) \\
  L_{\text{Inter-chain auxiliary F2E}} & = \frac{1}{N_{cycle} N_{pair}} \sum\limits_{c=0}^{N-1}\sum\limits_{
    i \in \{\text{backbone frames}\}, j \in \{\text{backbone frames}\},\text{chain(i)}\neq \text{chain(j)}}L_{\text{F2E}}(i, j, c) \\
\end{align}
When using the group-aware version of F2E, we only change the inter-chain part as
\begin{equation}
  L_{\text{Inter-chain auxiliary F2E, group-aware}} = \frac{1}{N_{cycle} N_{pair}} \sum\limits_{c=0}^{N-1}\sqrt{\sum\limits_{
    i \in \{\text{backbone frames}\}, j \in \{\text{backbone frames}\},\text{chain(i)}\neq \text{chain(j)}}L_{\text{F2E, group-aware}}(i, j, c)^2}. \\
\end{equation}

\subsection{Clamping Details}
The original FAPE loss is applied with clamping. Specifically, for some samples the loss is only applied to residue pairs within certain distance. In the early version of F2E, we apply the same restriction to the $R(3)$ part of the loss. However, later experiments show that removing the gradients of the $SO(3)$ at the same time can stabilize and accelerate the training procedure.

\subsection{Pytorch implementation of F2E}
We provide an example implementation of F2E in PyTorch 2.0.1.
\begin{python}
import torch
import einops
from typing import TypeAlias, Callable, Any

Tensor: TypeAlias = torch.Tensor

def invert_rigid(R: Tensor, t: Tensor):
    """Invert rigid transformation.

    Args:
        R: Rotation matrices, (..., 3, 3).
        t: Translation, (..., 3).

    Returns:
        R_inv: Inverted rotation matrices, (..., 3, 3).
        t_inv: Inverted translation, (..., 3).
    """
    R_inv = R.transpose(-1, -2)
    t_inv = -torch.einsum("... r t , ... t -> ... r", R_inv, t)
    return R_inv, t_inv

def node2pair(t1: Tensor, t2: Tensor, sequence_dim: int, op: Callable[[Tensor, Tensor], Tensor]) -> Tensor:
    """
    Create a pair tensor from a single tensor

    Args:
        t1: The first tensor to be converted to pair tensor
        t2: The second tensor to be converted to pair tensor
        sequence_dim: The dimension of the sequence
        op: The operation to be applied to the pair tensor

    Returns:
        Tensor: The pair tensor

    """
    # convert to positive if necessary
    if sequence_dim < 0:
        sequence_dim = t1.ndim + sequence_dim
    if t1.ndim != t2.ndim:
        raise ValueError(f"t1 and t2 must have the same number of dimensions, got {t1.ndim} and {t2.ndim}")
    t1 = t1.unsqueeze(sequence_dim + 1)
    t2 = t2.unsqueeze(sequence_dim)
    return op(t1, t2)

def compose_rotation_and_translation(
    R1: Tensor,
    t1: Tensor,
    R2: Tensor,
    t2: Tensor,
) -> tuple[Tensor, Tensor]:
    """Compose two frame updates.

    Ref AlphaFold2 Suppl 1.8 for details.

    Args:
        R1: Rotation of the first frames, (..., 3, 3).
        t1: Translation of the first frames, (..., 3).
        R2: Rotation of the second frames, (..., 3, 3).
        t2: Translation of the second frames, (..., 3).

    Returns:
        A tuple of new rotation and translation, (R_new, t_new).
        R_new: R1R2, (..., 3, 3).
        t_new: R1t2 + t1, (..., 3).
    """
    R_new = einops.einsum(R1, R2, "... r1 r2, ... r2 r3 -> ... r1 r3")  # (..., 3, 3)
    t_new = (
        einops.einsum(
            R1,
            t2,
            "... r t, ... t->... r",
        )
        + t1
    )  # (..., 3)

    return R_new, t_new

def masked_quadratic_mean(
    value: Tensor,
    mask: Tensor,
    dim: int | tuple[int, ...] | list[int] = -1,
    eps: float = 1e-10,
) -> Tensor | tuple[Tensor, Tensor]:
    """Compute quadratic mean value for tensor with mask.

    Args:
        value: Tensor to compute quadratic mean.
        mask: Mask of value, the same shape as `value`.
        dim: Dimension along which to compute quadratic mean.
        eps: Small number for numerical safety.
        return_masked: Whether to return masked value.

    Returns:
        Masked quadratic mean of `value`.
        [Optional] Masked value, the same shape as `value`.
    """
    return torch.sqrt((value * mask).sum(dim) / (mask.sum(dim) + eps))

def frame_aligned_frame_error_loss(
    R_pred: Tensor,
    t_pred: Tensor,
    R_gt: Tensor,
    t_gt: Tensor,
    frame_mask: Tensor,
    rotate_scale: float = 1.0,
    axis_scale: float = 20.0,
    eps_so3: float = 1e-7,
    eps_r3: float = 1e-4,
    dist_clamp: float | None = None,
    pair_mask: Tensor | None = None,
):
    """Compute frame aligned frame error loss with double geodesic metric.

    Args:
        R_pred: Predicted rotation matrices of frames, (..., N, 3, 3).
        t_pred: Predicted translations of frames, (..., N, 3).
        R_gt: Ground truth rotation matrices of frames, (..., N, 3, 3).
        t_gt: Ground truth translations of frames, (..., N, 3).
        frame_mask: Existing masks of ground truth frames, (..., N).
        axis_scale: Scale by which the R^3 part of loss is divided.
        eps_so3: Small number for numeric safety for arccos.
        eps_r3: Small number for numeric safety for sqrt.
        dist_clamp: Cutoff above which distance errors are disregarded.
        pair_mask: Additional pair masks of pairs which should be calculated, (..., N, M) or None.
            pair_mask=True, the FAPE loss is calculated; vice not calculated.
            If None, all pairs are calculated.

    Returns:
        Dict of (B) FAFE losses. Contains "fafe", "fafe_so3", "fafe_r3".
    """
    N = R_pred.shape[-3]

    def _diff_frame(R: Tensor, t: Tensor) -> Tensor:
        R_inv, t_inv = invert_rigid(
            R=einops.repeat(R, "... i r1 r2 -> ... (i j) r1 r2", j=N),
            t=einops.repeat(t, "... i t -> ... (i j) t", j=N),
        )
        R_j = einops.repeat(R, "... j r1 r2 -> ... (i j) r1 r2", i=N)
        t_j = einops.repeat(t, "... j t -> ... (i j) t", i=N)

        return compose_rotation_and_translation(R_inv, t_inv, R_j, t_j)

    frame_mask = node2pair(frame_mask, frame_mask, -1, torch.logical_and)
    if pair_mask is not None:
        frame_mask = pair_mask * frame_mask
    frame_mask = einops.rearrange(frame_mask, "... i j -> ... (i j)")

    losses = compute_double_geodesic_error(
        *_diff_frame(R_pred, t_pred),
        *_diff_frame(R_gt, t_gt),
        frame_mask=frame_mask,
        rotate_scale=rotate_scale,
        axis_scale=axis_scale,
        dist_clamp=dist_clamp,
        eps_so3=eps_so3,
        eps_r3=eps_r3,
    )
    return losses

def compute_double_geodesic_error(
    R_pred: Tensor,
    t_pred: Tensor,
    R_gt: Tensor,
    t_gt: Tensor,
    frame_mask: Tensor,
    rotate_scale: float = 1.0,
    axis_scale: float = 20.0,
    dist_clamp: float | None = None,
    eps_so3: float = 1e-7,
    eps_r3: float = 1e-4,
):
    """Compute frame-wise error with double geodesic metric.

    d_se3(T_pred, T_gt) = sqrt(d_so3(R_pred, R_gt)^2 + (d_r3(t_pred, t_gt) / axis_scale)^2)
    d_so3(R_pred, R_gt) range [0, pi]
    d_r3(t_pred, t_gt) / axis_scale) range [0, 1.5] when clamping
    
    Args:
        R_pred: Predicted rotation matrices of T, (..., N, 3, 3).
        t_pred: Predicted translations of T, (..., N, 3).
        R_gt: Ground truth rotation matrices of T, (..., N, 3, 3).
        t_gt: Ground truth translations of T, (..., N, 3).
        frame_mask: Existing masks of ground truth T, (..., N).
        rotate_scale: Scale by which the SO3 part of loss is divided.
        axis_scale: Scale by which the R^3 part of loss is divided.
        dist_clamp: Cutoff above which distance errors are disregarded.
        eps_so3: Small number for numeric safety for arccos.
            Refer to https://github.com/pytorch/pytorch/issues/8069
        ep3_r3: Small number for numeric safety for sqrt.

    Returns:
        Dict of (B) FAFE losses. Contains "fafe", "fafe_so3", "fafe_r3".

    Note:
        so3 loss/error presented in scaled form [0, pi/rotate_scale].
        r3 loss/error presented in scaled form [0, dist_clamp/axis_scale].
    """
    if dist_clamp is None:
        dist_clamp = 1e9

    # SO3 loss
    R_diff = einops.rearrange(R_pred, "... i j -> ... j i") @ R_gt
    R_diff_trace = R_diff.diagonal(dim1=-2, dim2=-1).sum(-1)
    so3_dist = torch.acos(torch.clamp((R_diff_trace - 1) / 2, -1 + eps_so3, 1 - eps_so3)) / rotate_scale  # (..., N)
    so3_loss = masked_quadratic_mean(so3_dist, frame_mask, dim=(-1))

    # R3 loss
    r3_dist = torch.sqrt(torch.sum((t_pred - t_gt) ** 2, dim=-1) + eps_r3)  # (..., N)
    r3_dist = r3_dist.clamp(max=dist_clamp) / axis_scale  # (..., N)
    r3_loss = masked_quadratic_mean(r3_dist, frame_mask, dim=(-1))

    # double geodesic loss
    se3_dist = torch.sqrt(so3_dist**2 + r3_dist**2)  # (..., N)
    se3_loss = masked_quadratic_mean(se3_dist, frame_mask, dim=(-1))

    losses = {
        "fafe": se3_loss,  # Note se3_loss = sqrt((so3_loss/rotate_scale)^2 + (r3_loss/axis_scale)^2)
        "fafe_so3": so3_loss,
        "fafe_r3": r3_loss,
    }

    return losses
\end{python}

\section{Experiment details} \label{appendix:experiment}
\subsection{Data preparation details}
\label{appendix:sub:detail_data}
 We crop the detected heavy/light chain and protein-type antigen from the full pdb structure. MSA is searched against these cropped structures. Next, we use ANARCI \cite{dunbar2016anarci} to annotate the loop region on VH/VL chains and further crop the antibody along with its MSA. We also notice there is a number of antigens with a very long sequence but only a limited region of ground-truth structures, so we crop the antigen to the longest sub-sequence that includes all the residues with experimental ground truth.

 After cropping, multiple filters are applied to clean data. SabDab entries without antigens or with non-protein antigens are discarded. All the chains are renamed according to the annotation by ANARCI, and samples with multiple H chains or multiple L chains detected are also discarded. There are also two special types of entries, nanobodies and Single-chain variable fragments (scFv). The nanobody is commonly annotated as a H-chain-only antibody, and scFvs are treated as two chains.

 For both evaluation sets, targets with small molecules, peptides or nucleic acids near the docking interface are removed. For low homology evaluation set, We further restrict to a subset of all the structures after January 12th, 2023, and manually check them in PDB website. Invalid targets like those with only partial antigen structure are removed. For full evaluation set, the total number is too large for manual checking, so we only check the cluster centers and find 12 invalid samples, 8 of which are either partial complex from a multiple-antigen or multiple-antibody assembly in which non-trivial interactions exist, 4 of which include unfolded chains in the experimental structure. Finally, for both evaluation sets, targets with full-complex sequence length larger than 1500 after cropping are removed for evaluation speed considerations.

\subsection{Training details}
\label{appendix:sub:detail_train}
For LoRA experiments, we apply LoRA only on 48-layer evoformer. All the parameters used before the evoformer are frozen. Structure module and all the prediction heads are trainable. In one version of our code we freeze the linear layer from node representations to single representations, which seems to harm the performance.

As the original paper~\cite{hu2022lora} suggests, we set LoRA $\alpha$ the same value as LoRA, and make a hyper parameter search for LoRA rank in [1, 4, 16, 32], and we found setting to 16 outperforms other values. The LoRA is added to all the linear layers in the evoformer, and for Q, K, V linear in attention, we add separate LoRA weights for each linear, despite their being fused during forwarding.

The training loss is set basically the same as original AF2. The batch size is set to 32, which we find is the minimal batch size to improve AF2 performance during fine-tuning. Training with larger batch size may improve performance, but we did not investigate it due to computing resource constraints. The cropping strategy is slightly different to original AF2 setting that during spatial cropping, only interface between antibody and antigen is considered. The crop length is set to 384 at max. The number of MSA cluster center is restricted to 128 for training efficiency. Cross-chain uniprot MSAs and templates are not provided, since we want to simplify our protocol and in theory antibody-antigen docking does not rely on these information. Finally, we use early stopping to avoid overfitting on training samples, which is a severe problem we have encountered in the early rounds of experiments.

\subsection{Inference details}
\label{appendix:sub:detail_inference}
During the inference, the random seed is set to 0 which will be applied to all data transformations. This means for all the model the MSA clustering results for the same recycle in different models will be controlled as same. We think this is a more reasonable settings than randomly sampling MSAs for different models, because we observe the sampling procedure can introduce large bias, sometimes making a sample rotates to correct poses from wrong ones for the same model. 

\section{Additional results} \label{appendix:results}
\begin{figure*}[htbp]
% \vskip -0.2in
\begin{center}
\centerline{\includegraphics[width=0.5\textwidth]{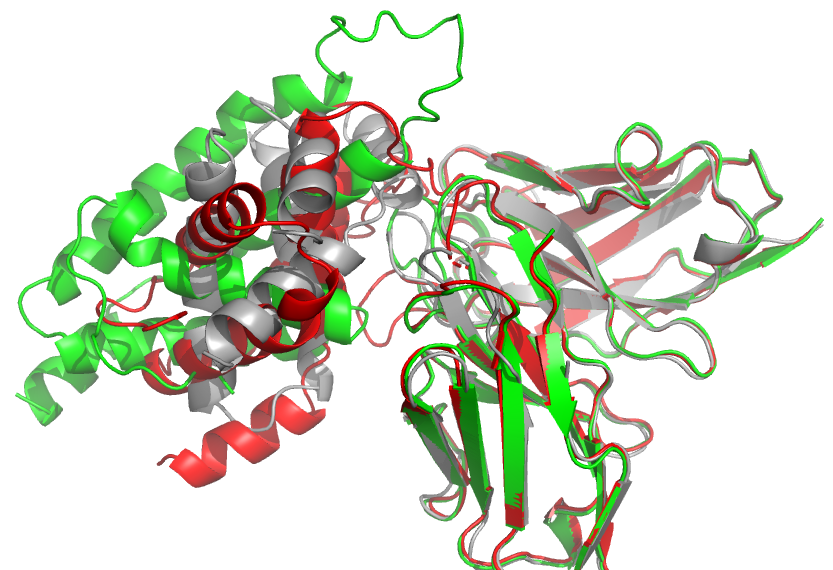}}
\vskip -0.1in
\caption{
% Case study and error analysis. RW: remove for space saving
The predicted structures of PDB 8DB4 against ground truth. The grey structure is the experimental ground truth, the green structure by AF2.3 and the red structure by model tuned with F2E.
}
\end{center}
\vskip -0.3in
\end{figure*}

\begin{figure*}[htbp]
% \vskip -0.2in
\begin{center}
\centerline{\includegraphics[width=0.5\textwidth]{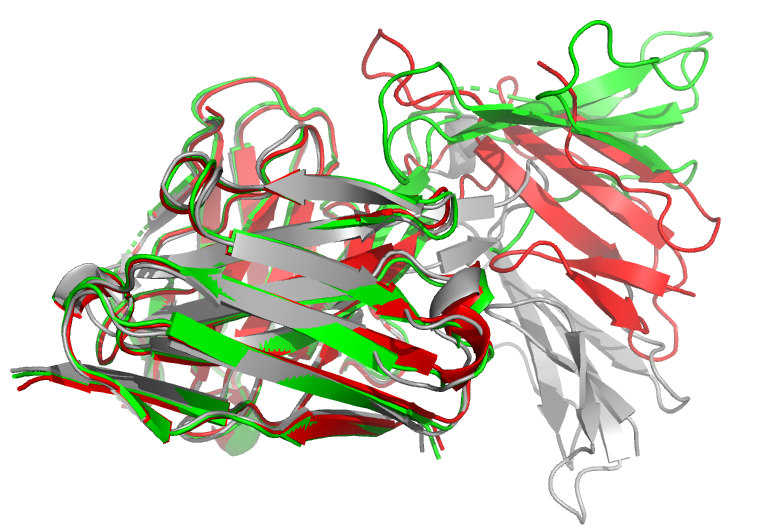}}
\vskip -0.1in
\caption{
% Case study and error analysis. RW: remove for space saving
The predicted structures of PDB 8D9Y against ground truth. The grey structure is the experimental ground truth, the green structure by AF2.3 and the red structure by model tuned with F2E.
}
\end{center}
\vskip -0.3in
\end{figure*}

\begin{figure*}[htbp]
% \vskip -0.2in
\begin{center}
\centerline{\includegraphics[width=0.5\textwidth]{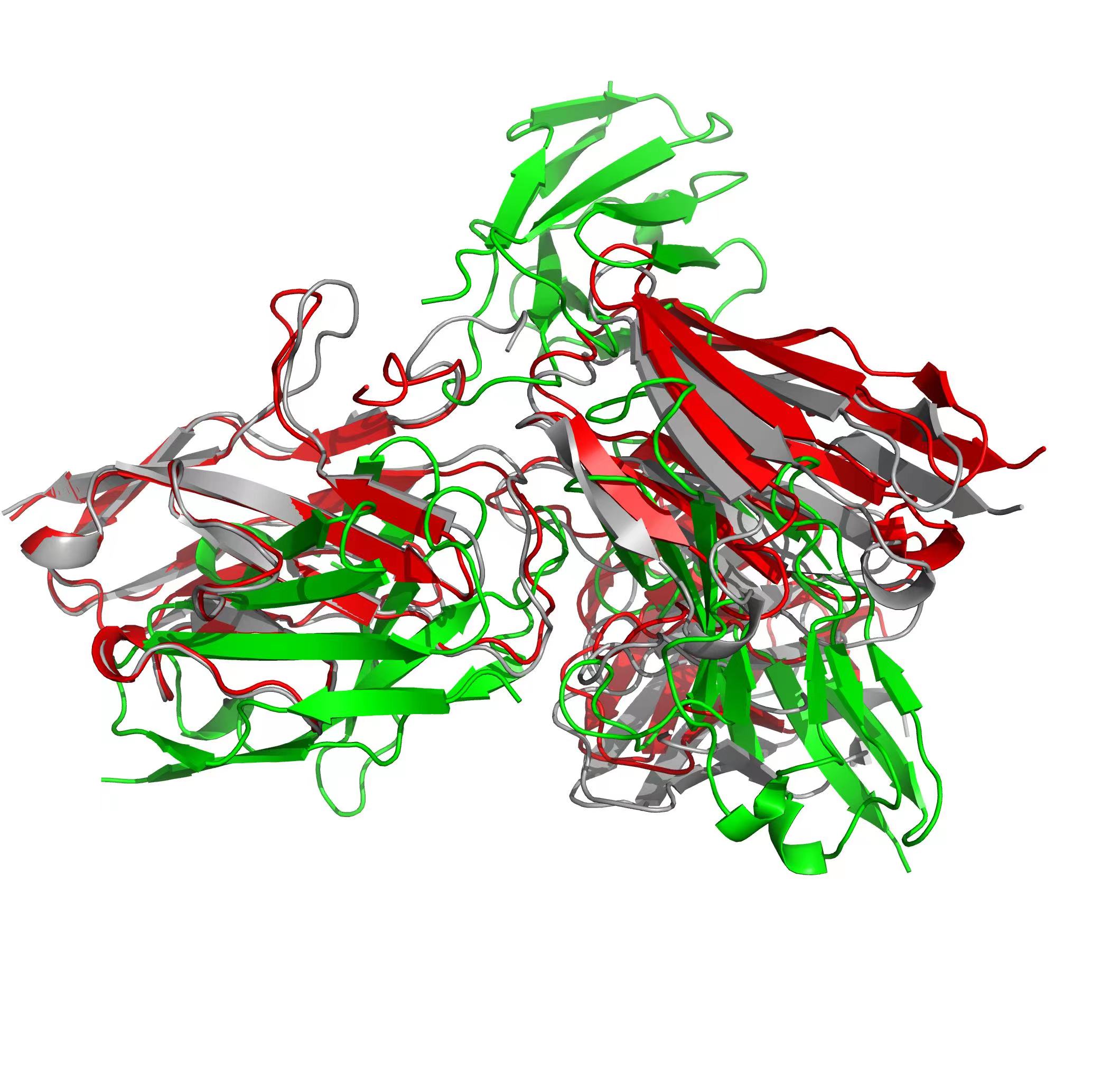}}
\vskip -0.1in
\caption{
% Case study and error analysis. RW: remove for space saving
The predicted structures of PDB 7QUH against ground truth. The grey structure is the experimental ground truth, the green structure by AF2.3 and the red structure by model tuned with F2E.
}
\end{center}
\vskip -0.3in
\end{figure*}

\begin{figure*}[htbp]
% \vskip -0.2in
\begin{center}
\centerline{\includegraphics[width=0.5\textwidth]{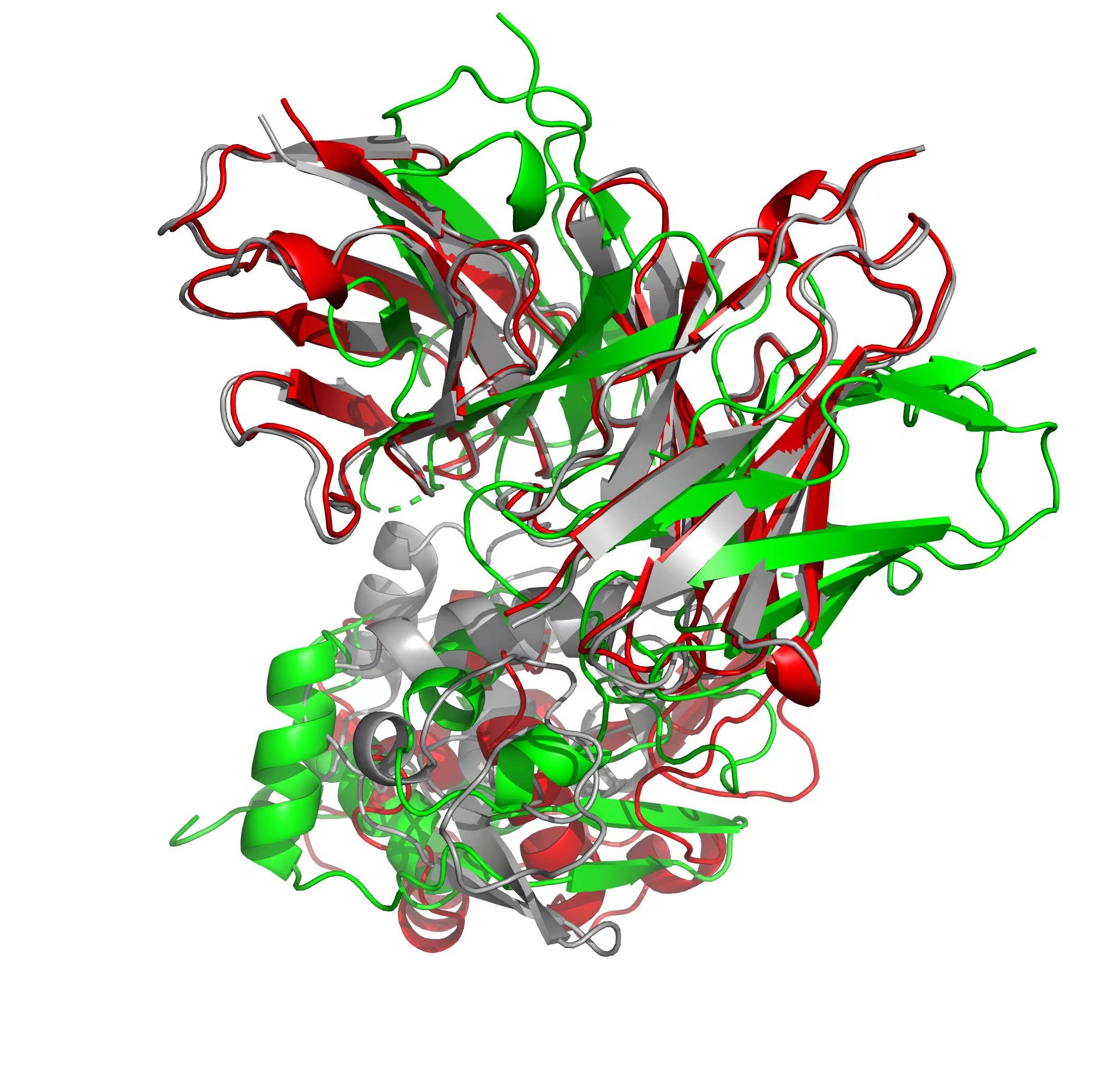}}
\vskip -0.1in
\caption{
% Case study and error analysis. RW: remove for space saving
The predicted structures of PDB 8GQ1 against ground truth. The grey structure is the experimental ground truth, the green structure by AF2.3 and the red structure by model tuned with F2E.
}
\end{center}
\vskip -0.3in
\end{figure*}

%%%%%%%%%%%%%%%%%%%%%%%%%%%%%%%%%%%%%%%%%%%%%%%%%%%%%%%%%%%%%%%%%%%%%%%%%%%%%%%
%%%%%%%%%%%%%%%%%%%%%%%%%%%%%%%%%%%%%%%%%%%%%%%%%%%%%%%%%%%%%%%%%%%%%%%%%%%%%%%

\end{document}